\documentclass[a4paper,12pt]{amsart}
 
\usepackage{amssymb, amsmath, amsthm, amsfonts}
\usepackage{extsizes}
\usepackage{hyperref}
\usepackage{xcolor,graphicx}

\theoremstyle{plain}
\newtheorem{theorem}{Theorem}[section]

\newtheorem{proposition}[theorem]{Proposition}

\theoremstyle{definition}
\newtheorem{definition}[theorem]{Definition}

\newtheorem{example}[theorem]{Example}

\def\r{\mathbb R}

 \def\s{\mathbb S}
 \def\S{\mathcal S}
\usepackage{xcolor,graphicx}

\begin{document}

\title[The radial Newton problem]{The radial Newton problem: nonlinear dynamics of minimal resistance in central fields}

\author[R. L\'opez]{Rafael L\'opez}
\address{Rafael L\'opez \\ Department of Geometry and Topology\\ University of  Granada. 18071 Granada. Spain}
\email{rcamino@ugr.es}

\begin{abstract}
This paper investigates the nonlinear dynamics  of Newton's problem of minimal resistance in radial fields. We move beyond classical translational symmetry to analyze two non-equilibrium scenarios: a scale-invariant free expansion and an incompressible source flow. Our analysis reveals that the scale-invariant model suffers from a symmetry-breaking instability (loss of ellipticity) that necessitates geometric truncation. Conversely, we prove that the incompressible flow acts as a structural regularizer, admitting unique, smooth, and strictly concave solutions. These findings provide new qualitative insights into how physical conservation laws ensure the regularity and symmetry of optimal configurations in high-speed central flows, bridging the gap between variational calculus and the physics of complex systems.\end{abstract}

\maketitle

%%%%%%%%%%%%%%%%%%%%%%
\section{Introduction and statement of the problem}
%%%%%%%%%%%%%%%%%%%%%%

One of the most formidable challenges in modern aerospace engineering is the geometric optimization of spacecraft capsules, particularly during the two critical phases of a mission: atmospheric launch and, most crucially, planetary reentry. During reentry, vehicles plunge into the atmosphere at hypersonic speeds, triggering intense non-equilibrium processes where vast amounts of kinetic energy are violently dissipated into pressure and heat.  As seen in the Orion (Artemis II) spacecraft, the blunt-body geometry induces a detached bow shock that acts as a nonlinear interface, resulting in a flow pattern that radiates from the shock front (Fig. \ref{fig0}, left.). See \cite{an,miele}. The crew module of the Orion spacecraft,   provides a contemporary realization of a classical problem in Newtonian aerodynamics. Its external geometry can be approximated  as a truncated cone (conical frustum), although the actual configuration is more accurately described as a blunted cone with smoothly varying curvature.  The module has a base diameter of approximately $5.03\,\mathrm{m}$, corresponding to the thermal protection system, and an overall height of about $3.3\,\mathrm{m}$. The effective half-angle of the conical section lies in the range $30^\circ$--$33^\circ$, closely reflecting the legacy of the Apollo command module design \cite{an,NASAOrion}.

From a mathematical standpoint, this geometry embodies a nontrivial compromise between competing aerodynamic and thermodynamic constraints. The adoption of a blunted cone geometry is a fundamental thermal management strategy in hypersonic reentry; by replacing a sharp-tipped frustum with a smooth, rounded nose, the resulting detached bow shock creates a standoff distance that significantly reduces the convective heat flux to the vehicle's surface \cite{ae,bertin,ha}. 
Consequently, the Orion capsule may be interpreted as an optimized solution, within modern engineering constraints, to the problem of determining the body of least resistance formulated by Isaac Newton in 1687 \cite{newton}. 

However, such optimizations are intrinsically linked to the underlying structure of the flow field.  Newton's Principia relies on the fundamental assumption that the particle flow is strictly uniform and parallel \cite{buttazzo, lachand, plakhov_book}. While this translational symmetry serves as a highly effective local approximation for small-scale scenarios, it must be fundamentally revised for planetary scales where the very nature of the trajectory alters the resistance profile.

During reentry from orbital or lunar velocities, the incoming flow is governed by the Earth's central gravitational field, which induces a radially varying acceleration. This creates a shift from a simplified translational regime to a complex nonlinear dynamics framework, where the particle trajectories and kinematic forces are intrinsically radial  \cite{an,regan1993,bertin}. Consequently, at these macroscopic scales, the geometry of the medium dictates a curved, trajectory-dependent velocity field, rendering the resistance problem fundamentally central \cite{be, ca, pa, se}.

Therefore, the geometry of capsules like Orion is the outcome of an optimization process that implicitly accounts for these  non-parallel flow effects and symmetry-breaking phenomena which lie beyond the scope of classical models. Extending Newton's problem to a radial vector field is not merely an abstract mathematical generalization, but a necessary step towards a rigorous physical model capable of capturing the emergent features of high-speed flow in central fields. 

 \begin{figure}
 \includegraphics[width=.35\textwidth]{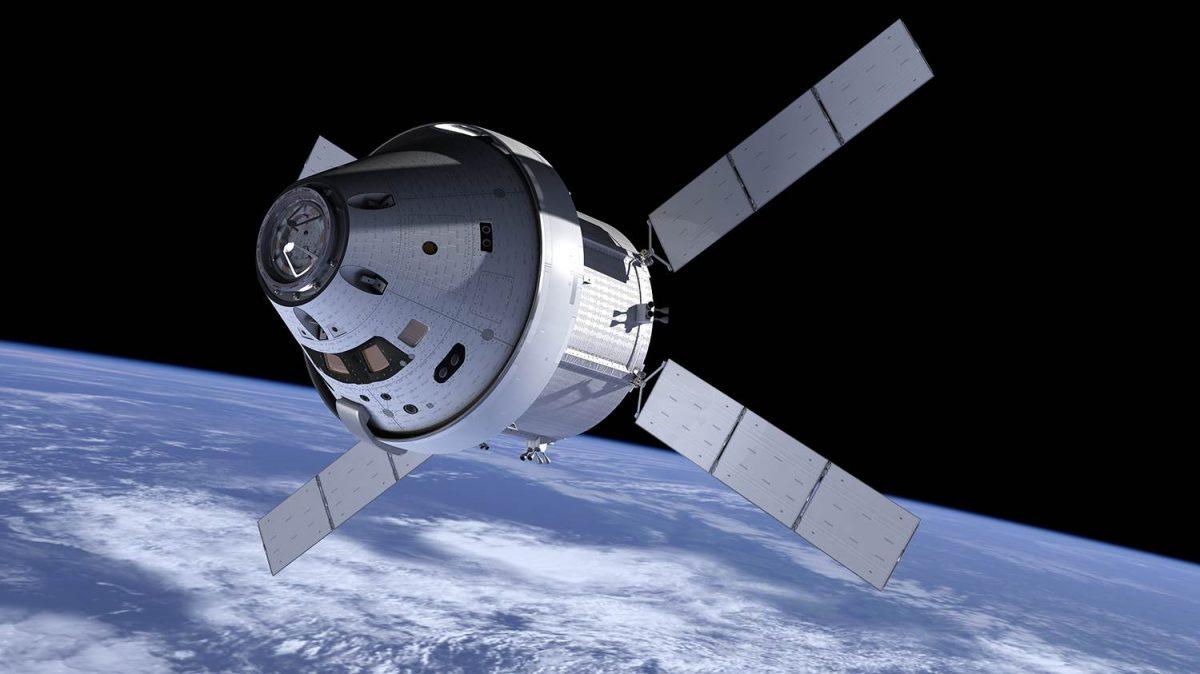} \includegraphics[width=.25\textwidth]{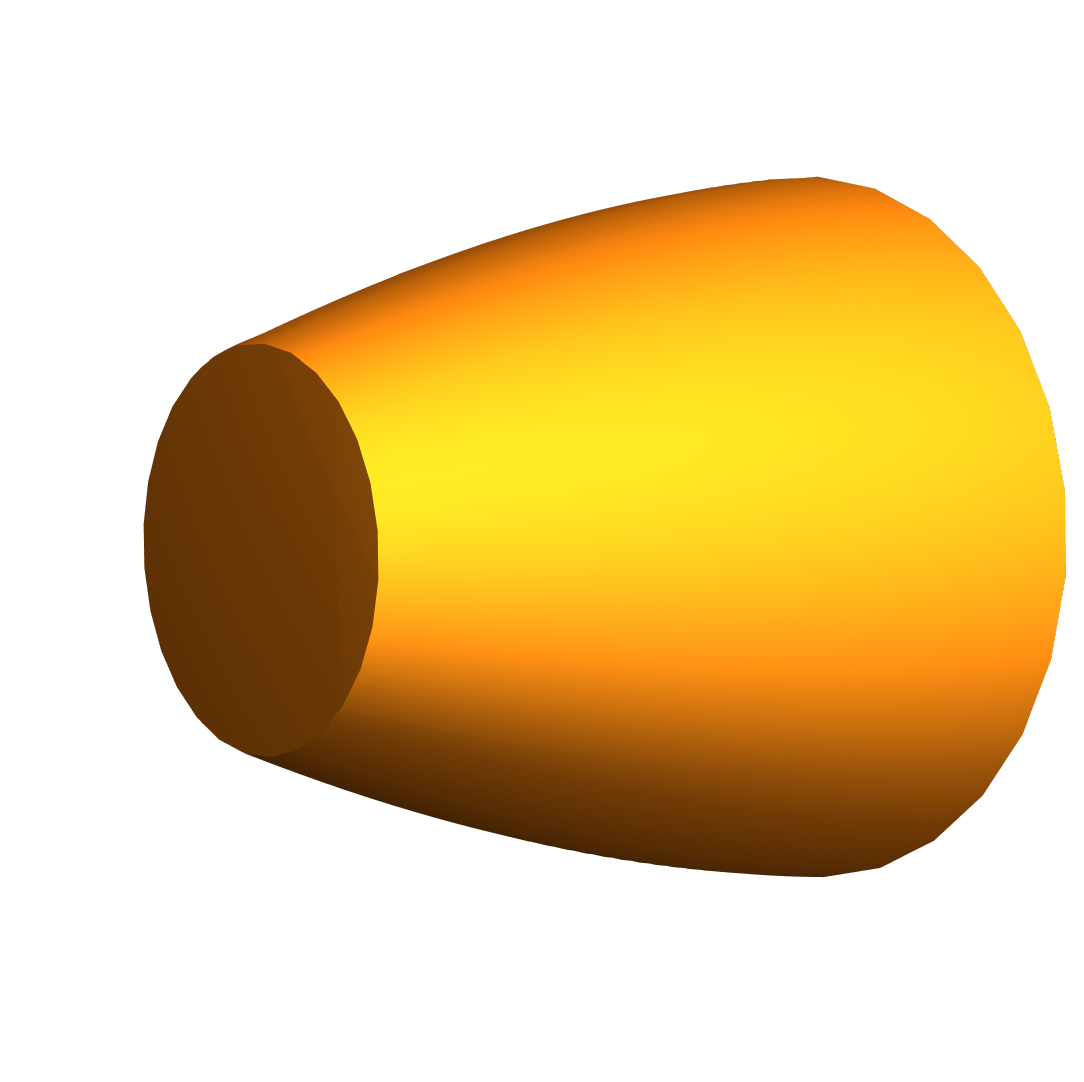} \quad\includegraphics[width=.12\textwidth]{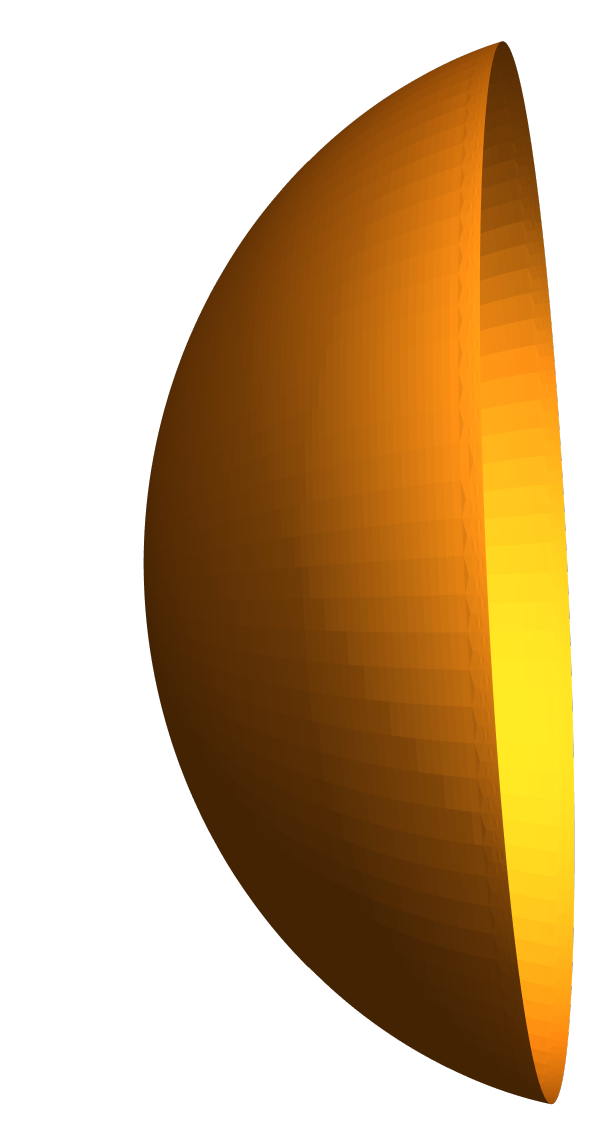}
 \caption{Left: The Orion capsule, illustrating a geometry governed by a frustum cone design. Middle and right: The optimal radial frustum cone and the smooth radially symmetric stationary configuration   derived from the two models proposed in this work. }\label{fig0}
 \end{figure}

In the present work, we assume that the particles are emitted from a central point source located at the origin of $\mathbb{R}^3$ and travel outward along radial trajectories. The problem can be formulated as follows:

\begin{quote}
{\it Problem:} Find the optimal geometric configuration of a solid body  that minimizes the total resistance offered to a radial particle flow emanating from the origin of Euclidean space.
\end{quote}

To translate this physical setup into a formal mathematical model, we consider a solid body $\mathcal{B}$ in $\mathbb{R}^3$ bounded by a piecewise $C^1$ surface $\Sigma$. By placing the source of the flow at the origin $O=(0,0,0)$, we assume that the body is star-shaped with respect to this center. This geometric assumption is natural for objects interacting with divergent flows, as it guarantees that $\Sigma$ can be parameterized as a radial graph over a domain $\Omega$ of the unit sphere $\mathbb{S}^2$. Consequently, in spherical coordinates, any point $\mathbf{x} \in \Sigma$ is defined by the position vector $\mathbf{x}(\xi) = \rho(\xi)\xi$, where $\xi \in \Omega$ represents the directional unit vector and the function $\rho: \Omega \to \mathbb{R}^+$ determines the radial profile of the vehicle.

To model the interaction between the object and the central field, we consider a radial flow of particles in $\mathbb{R}^3$ emanating from (or converging to) the origin. The kinematics of this medium are governed by a purely radial velocity field $\vec{V}(\mathbf{x}) = v(\rho) \hat{e}_r(\mathbf{x})$, where $\hat{e}_r(\mathbf{x})= \mathbf{x}/\rho$ denotes the unit radial vector. The force exerted by the flow is evaluated using the classical Newtonian impact model \cite{buttazzo,go,plakhov_book}. In this regime, when a particle  collides with the surface $\Sigma$, it undergoes a perfectly inelastic collision, transferring only its normal momentum to the body. We also impose the standard single-impact assumption, meaning that after a particle collides with $\Sigma$, it does not strike the surface a second time.  The differential resistance force $d\vec{R}$ acting on a surface element  $d\Sigma$ at  a point $\mathbf{x} \in \Sigma$ is proportional to the local particle flux density $\Phi(\rho)$, the square of the normal component of the velocity, and the geometry of the impact:
\[ d\vec{R} = \Phi(\rho) \langle \vec{V}, \hat{n} \rangle^2 \, d\Sigma\, \hat{n},\]
where $\hat{n}$ is   the unit outward normal to the surface.

The formulation of this radial problem depends critically on the nonlinear coupling between the geometry of the body and the kinematic properties of the flow. By invoking fundamental conservation laws, we identify two distinct physical scenarios that yield different stationary configurations:  

\begin{enumerate}
\item {\it The scale-invariant free expansion model.} This scenario models a free expansion in a vacuum, analogous to the asymptotic regime of a stellar wind or the coasting phase of an isotropic explosion \cite{pa, pla}. In this regime, particles travel at a constant velocity $v(\rho) = v_0$. According to the conservation of mass, the number of particles crossing any spherical shell remains constant. Since the area of these shells grows as $\rho^2$, the particle flux density $\Phi(\rho)$ must decay inversely with the square of the radial distance, yielding $\Phi(\rho)=1/\rho^2$. Because the area element $d\Sigma$ of the radial graph scales proportionally with $\rho^2$, this geometric spreading exactly compensates for the flux attenuation, leading to a scale-invariant functional that, as we shall demonstrate, exhibits a symmetry-breaking instability (loss of ellipticity):
\begin{equation}\label{m1}
 R[\Sigma] = \int_{\Omega} \frac{\rho^2}{\rho^2 + |\nabla_{\mathbb{S}^2} \rho|^2} \, d\Omega,
\end{equation}

\item {\it The incompressible source flow model.} This model considers a point source emitting a continuous flow into a saturated, incompressible medium, such as a hydrodynamic source or a well in a porous medium \cite{batchelor2000, be,la}. In this steady state, the volumetric flow rate is constant across any spherical shell, forcing the radial velocity to decay as $v(\rho) = 1/\rho^2$. The momentum transfer is dominated by the dynamic pressure, which scales with the square of the velocity, resulting $v(\rho)^2 = 1/\rho^4$. Consequently, the resistance functional   takes the form:
\begin{equation}\label{m2}
 R[\Sigma] = \int_{\Omega} \frac{1}{\rho^2 + |\nabla_{\mathbb{S}^2} \rho|^2} \, d\Omega.
\end{equation}
 In this case, the mass-conservation law acts as a physical regularizer, yielding a functional that admits smooth closed, and  rotational symmetric stationary configurations.
 
\end{enumerate}

The primary aim of this work is to introduce and formalize, to the best of the author's knowledge for the first time in the literature, the radial version of Newton's problem of minimal resistance. In the classical parallel-flow setting, the search for global minimizers is already an exceptionally subtle task that often involves the loss of convexity or symmetry \cite{buttazzo, lachand}. Given the added complexity of radial geometry and distance-dependent flux, a full treatment of global optimality exceeds the scope of this initial study. Instead, our objective is to establish a robust mathematical and  physical framework by deriving the corresponding resistance functionals, obtaining their governing nonlinear equations, and analyzing the geometric nature of the resulting  stationary configurations (or extremal in the language of calculus of variations). By formalizing these models, we provide a starting point for future investigations in shape optimization under central force fields. As illustrated in Fig. \ref{fig0}, our results bridge the gap between engineering intuition and formal variational analysis, ranging from the  radial frustum cone (middle) to a novel, smooth, and  radially symmetric configuration (right) that emerges when physical regularizers are considered.

The paper is organized as follows. Sections \ref{s2} and \ref{s3} deal with the scale-invariant expansion and the incompressible source flow models, respectively. Each section includes the derivation of the physical model, the computation of the resistance for prototypical shapes,  the formulation of the Euler-Lagrange equation, the analysis of minimizers of radial frustum cones under height constraints and, finally, the study of rotationally symmetric stationary configurations. In the final Section \ref{s5}, we offer some concluding remarks and suggest open problems for future research in this new  paradigm.

%%%%%%%%%%%%%%%%%%%%%%%%%%%%%%
\section{The scale-invariant free expansion}\label{s2}
%%%%%%%%%%%%%%%%%%%%%%%%%%%%%%%%%%
 %%%%%%
\subsection{Derivation of the model}
%%%%%

We parametrize the boundary surface $\Sigma=\partial\mathcal{B}$ as a radial graph over a domain $\Omega \subset \mathbb{S}^2$. Let $\rho\colon \Omega \to \mathbb{R}^+$ be a strictly positive scalar function representing the radial distance from $O$. A point $\mathbf{x}\in \Sigma$ is given by the position vector
\begin{equation}
 \mathbf{x}= \rho(\xi) \hat{e}_r( \mathbf{x}), \quad \xi \in \Omega,
\end{equation}
where $\hat{e}_r( \mathbf{x})=\frac{\mathbf{x}}{\rho}=\xi$ denotes the outward-pointing radial unit vector at the angular position $\xi$. 
The unit normal vector $\hat{n}$ to $\Sigma$ can be expressed using the surface gradient on the sphere, $\nabla_{\mathbb{S}^2} \rho$, as
\begin{equation*}
 \hat{n} = \frac{\rho \hat{e}_r - \nabla_{\mathbb{S}^2} \rho}{\sqrt{\rho^2 + |\nabla_{\mathbb{S}^2} \rho|^2}}.
\end{equation*}
The incident particles follow the trajectory $\hat{e}_r$. Consequently, the angle of incidence $\alpha$ between the flow direction and the outward normal vector to $\Sigma$ is given by
\begin{equation}\label{cc}
 \cos \alpha = \langle \hat{e}_r , \hat{n}\rangle = \frac{\rho}{\sqrt{\rho^2 + |\nabla_{\mathbb{S}^2} \rho|^2}}.
\end{equation}
In our radial model, we  assume a flow of non-interacting particles moving strictly in the outward radial direction $\hat{e}_r$. Let $\vec{\Phi}=\Phi(\rho)\hat{e}_r$ denote the particle flux density vector, representing the number of particles crossing a unit area per unit time. Multiplying by the area element $4\pi\rho^2$, the total number of particles crossing the sphere is $4\pi\rho^2 \vec{\Phi}$. Conservation of mass requires that the divergence of the flux density vector vanishes; that is, $\mbox{div}(\rho^2\vec{\Phi})= 0$. In spherical coordinates, since $\vec{\Phi}$ depends only on the radial distance and points solely in the radial direction, this condition becomes 
 $$ \frac{\partial}{\partial \rho} \left( \rho^2 \Phi(\rho) \right) = 0.$$
 Thus, up to a constant, we have $\Phi(\rho)=\frac{1}{\rho^2}$. Physically, this represents the spatial dispersion from a constant source (the origin of $\r^3$) which emits a constant number of particles $N$ per unit time into a specific solid angle corresponding to $\Omega$. The total area of the spherical cap crossed by these particles at a distance $\rho$ is $\mbox{area}( \Omega) \rho^2$. To conserve the total number of particles crossing any such surface, the flux density must scale inversely with the area, i.e., $N/ (\mbox{area}(\Omega) \rho^2)$.

Let $d\Sigma$ be a differential area element on $\Sigma$, and $d\Omega$ the corresponding area element on the unit sphere. Both geometric quantities are related by the fact that the projection of $d\Sigma$ onto the sphere, divided by $\rho^2$ coincides with $d\Omega$, that is,  
\begin{equation}\label{8}
 d\Omega=\frac{\langle\hat{e}_r,\hat{n} \rangle\, d\Sigma}{\rho^2}= \frac{\cos\alpha \, d\Sigma }{\rho^2}.
\end{equation}
According to Newton's impact theory, the local normal force $d\vec{F}_n$ exerted on $d\Sigma$ is proportional to the number of particles striking the surface per unit time, $\Phi(\rho) v_0\cos \alpha \, d\Sigma$, and the normal component of the momentum transfer ($\langle \hat{e}_r, \hat{n}\rangle = \cos \alpha$). Thus, up to proportionality, 
$$d\vec{F}_n =\Phi(\rho) \cos^2 \alpha \, d\Sigma \, \hat{n}.$$
 The resistance element $dR$ is the projection of this force in the direction of the flow $\hat{e}_r$, hence, 
$$dR=\langle d\vec{F}_n,\hat{e}_r\rangle=\Phi(\rho) \cos^2 \alpha \langle\hat{n},\hat{e}_r\rangle\, d\Sigma=\Phi(\rho) \cos^3 \alpha \, d\Sigma.$$
 Using \eqref{8}, the total resistance is the integral of $dR$ over $\Sigma$, yielding
\begin{equation*}
 R[\Sigma] = \int_{\partial\mathcal{B}} \frac{1}{\rho^2} \cos^3 \alpha \, d\Sigma = \int_{\Omega} \cos^2 \alpha \, d\Omega.
\end{equation*}
Substituting the expression for $\cos \alpha$ in \eqref{cc}, we obtain the functional
\begin{equation*}
 R[\Sigma] = \int_{\Omega} \frac{\rho^2}{\rho^2 + |\nabla_{\mathbb{S}^2} \rho|^2} \, d\Omega,
\end{equation*}
which coincides with \eqref{m1}. It is preferable to introduce the logarithmic transformation 
$$u = \log \rho,\qquad \Sigma=\{ e^u\cdot\xi\colon \xi\in\Omega\}.$$ 
 Since $|\nabla_{\mathbb{S}^2} \rho|^2=\rho^2 |\nabla_{\mathbb{S}^2} u|^2$, we give the following definition: 

\begin{definition} 
Given a surface $\Sigma$ as the radial graph $e^u$ over a spherical domain $\Omega\subset\s^2$, its resistance $R[u]$, also denoted by $R[\Sigma]$, is 
\begin{equation} \label{m11}
 R[u] = \int_{\Omega} \frac{1}{1 + |\nabla_{\mathbb{S}^2} u|^2} \, d\Omega.
\end{equation}
\end{definition}
 The functional $R[u]$ in \eqref{m11} is formally identical to that of the classical model. A fundamental property of the classical Newton problem is its translational invariance along the direction of the flow. In the scale-invariant free expansion model, this symmetry is replaced by a homothetic invariance. If we dilate $\mathcal{B}$ by a constant factor $\lambda > 0$, the new radial profile is $\tilde{\rho} = \lambda \rho$, which corresponds to $\tilde{u} = u + \log \lambda$. It is then clear that the resistance functional \eqref{m11} remains invariant. 

%%%%
\subsection{Resistance of  canonical  configurations}
%%%%%

We first evaluate the consistency of the radial functional by computing the resistance of several canonical bodies. To do so, let us introduce spherical coordinates on $\s^2$, which will be used throughout this paper. We parametrize $\s^2$ by 
$$\Psi(\theta,\phi)=(\sin\theta\cos\phi,\sin\theta\sin\phi,\cos\theta),$$
 where $\theta \in [0, \pi]$ is the colatitude angle and $\phi \in [0, 2\pi)$ is the azimuthal angle. We express \eqref{m11} in spherical coordinates. The metric tensor $g$ on $\s^2$ has components $g_{\theta\theta} = 1$, $g_{\phi\phi} = \sin^2\theta$, and $g_{\theta\phi} = 0$, with determinant $g = \sin^2\theta$. The squared norm of the surface gradient is given by
\begin{equation}\label{uu}
 |\nabla_{\mathbb{S}^2} u|^2 = u_\theta^2 + \frac{1}{\sin^2\theta} u_\phi^2,
\end{equation}
where $u_\theta = \partial u/\partial \theta$ and $u_\phi = \partial u/\partial \phi$. Thus \eqref{m11} becomes
\begin{equation}\label{r7}
R[u]=\int\int \frac{\sin^3\theta}{(1+u_\theta^2)\sin^2\theta+u_\phi^2}\, d\theta\, d\phi.
\end{equation}
If $u$ is rotationally symmetric about the $z$-axis, then $u(\theta,\phi)=u(\theta)$. Then $u_\phi=0$, $u_\theta=u'$ and 
$$R[u]=\int\int \frac{\sin\theta}{1+u'^2}\, d\theta\, d\phi.$$
To understand the nonlinear dynamics of the scale-invariant free expansion, we first evaluate the energy landscape of the basic radial surfaces. See Figure \ref{fig1}.

\begin{example} The spherical cap. A spherical cap $\S_R$ of radius $R\in (0,\pi)$ is the radial graph of $u=0$, $\S_R=\Psi( [0,R]\times[0,2\pi])$. Then, the resistance reduces to the area of the cap $\S_R$, 
\begin{equation}\label{rrr}
R[\S_R]=2\pi(1-\cos R).
\end{equation}
\end{example}

\begin{example} The flat disk. A flat disk $\mathcal{D}_R$ of radius $\sin R$, $R\in (0,\pi)$, is the horizontal disk at height $z=\cos R$ whose boundary is $\partial\S_R$. This disk is the radial graph of $\rho(\theta)=\frac{\cos R}{\cos\theta}$, or equivalently, $u(\theta)=\log(\cos R)-\log (\cos\theta)$. Its resistance is 
$$R[\mathcal{D}_R]=2\pi\frac{1-\cos^3 R}{3}.$$
\end{example}

\begin{example}\label{ex24}
 The radial cone. A radial cone is the radial graph of a linear function of the $\theta$-variable. Let $R\in (0,\pi)$. We define the radial cone $\mathcal{C}_R$ of radius $R$ and slope $k>0$ as the radial graph of the function $u(\theta) = k(R - \theta)$, $(\theta,\phi)\in [0,R]\times [0,2\pi]$. The boundary of $\mathcal{C}_R$ is the circle $\theta=R$. The resistance of $\mathcal{C}_R$ is 
\begin{equation}\label{r3}
 R[\mathcal{C}_R]  =\frac{2\pi}{1+k^2} (1 - \cos R).
\end{equation}
 It is noteworthy that $R[\mathcal{C}_R] = (1+k^2)^{-1} R[\S_R]$. 
\end{example}

Before   continuing, we motivate the following example by drawing a parallel with classical theory. In the classical model, Newton observed that the ratio  between the resistance of a disk perpendicular to the flow and that of a sphere, both with the same radius, is
$$\frac{R[\textrm{sphere}]}{R[\textrm{disk}]}=\frac12.$$
This ratio is independent of the radius. We study the analogy in the radial case. In our framework, the disk is substituted by a spherical cap $\S_R$ of radius $R$ because, in both settings, the surfaces correspond to constant functions. The analogue of the sphere is the subdomain of a sphere centered at the $z$-axis that is tangent to the cone from $O$ determined by $\S_R$.

\begin{example} 
The tangent spherical cap. We consider a spherical cap $\S_{tan}$ whose boundary coincides with $\partial \S_R$ and which is tangent to the radial vector field along the boundary curve $\partial \S_R$. In fact, there are two (distinct) tangent spherical caps that share the same boundary. In radial coordinates, the tangent spherical caps are given by 
$$\rho_{\pm}(\theta)=\frac{1}{\cos R}\left(\cos\theta \pm\sqrt{\cos^2\theta-\cos^2 R}\right),$$
where $+$ and $-$ indicate the upper and the lower tangent spherical caps, respectively. 
A computation gives 
$$R[\rho_{\pm}]=2\pi\frac{\cos R-\cos(2R)}{3(1+\cos R)}.$$
Note that the same resistance is offered by both tangent spherical caps.
Thus, the ratio is 
$$\frac{R[\S_{tan}]}{R[\S_R]}=\frac{1+2\cos R}{3(1+\cos R)}.$$
Although this is not constant, it is not difficult to see that 
$$\lim_{R\to 0}\frac{R[\S_{tan}]}{R[\S_R]}=\frac12.$$
This result elegantly demonstrates that the classical Newtonian limit emerges as a local approximation of our global radial field. 
\end{example}

 \begin{figure}
 \includegraphics[width=.22\textwidth]{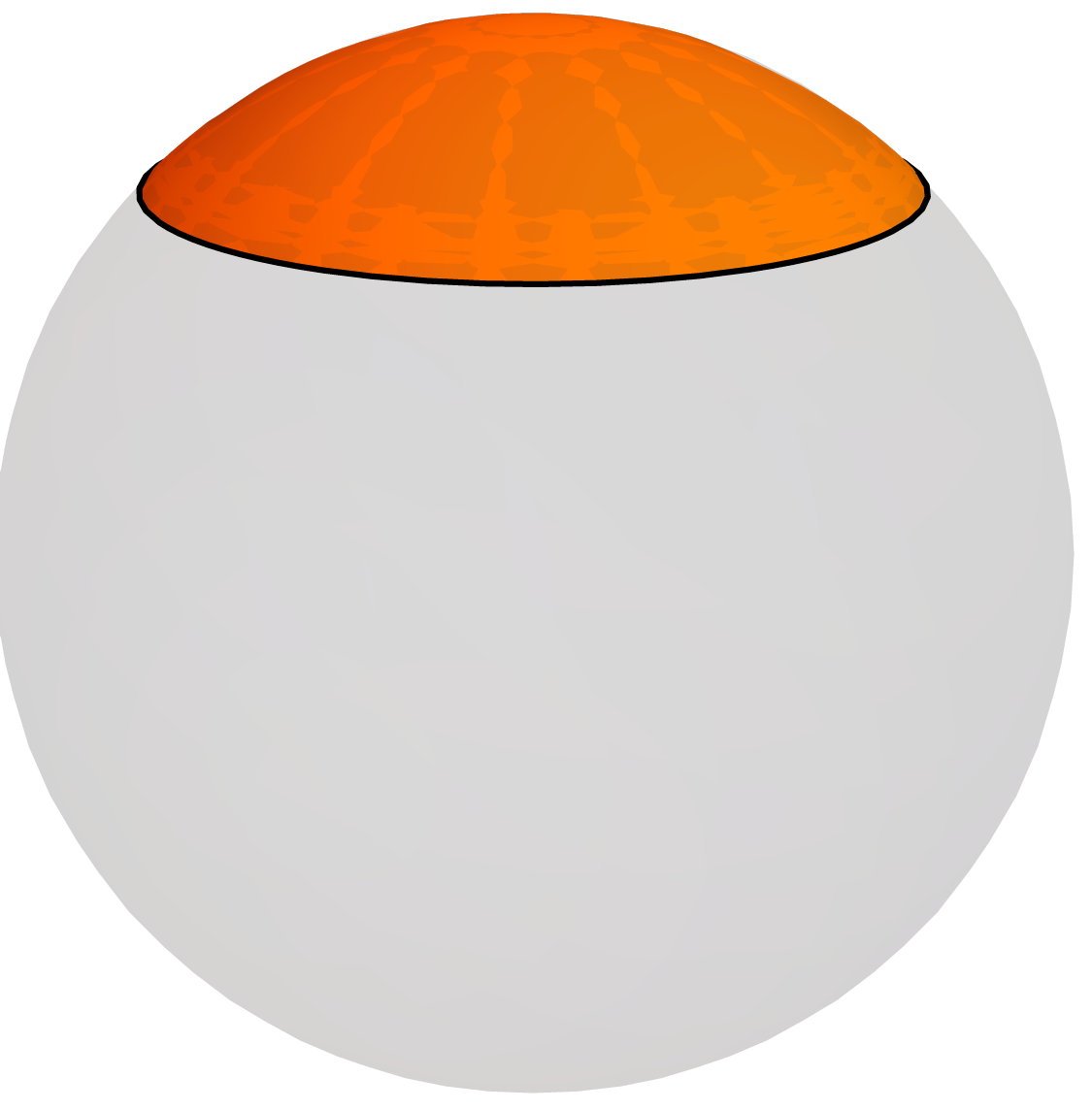}\includegraphics[width=.22\textwidth]{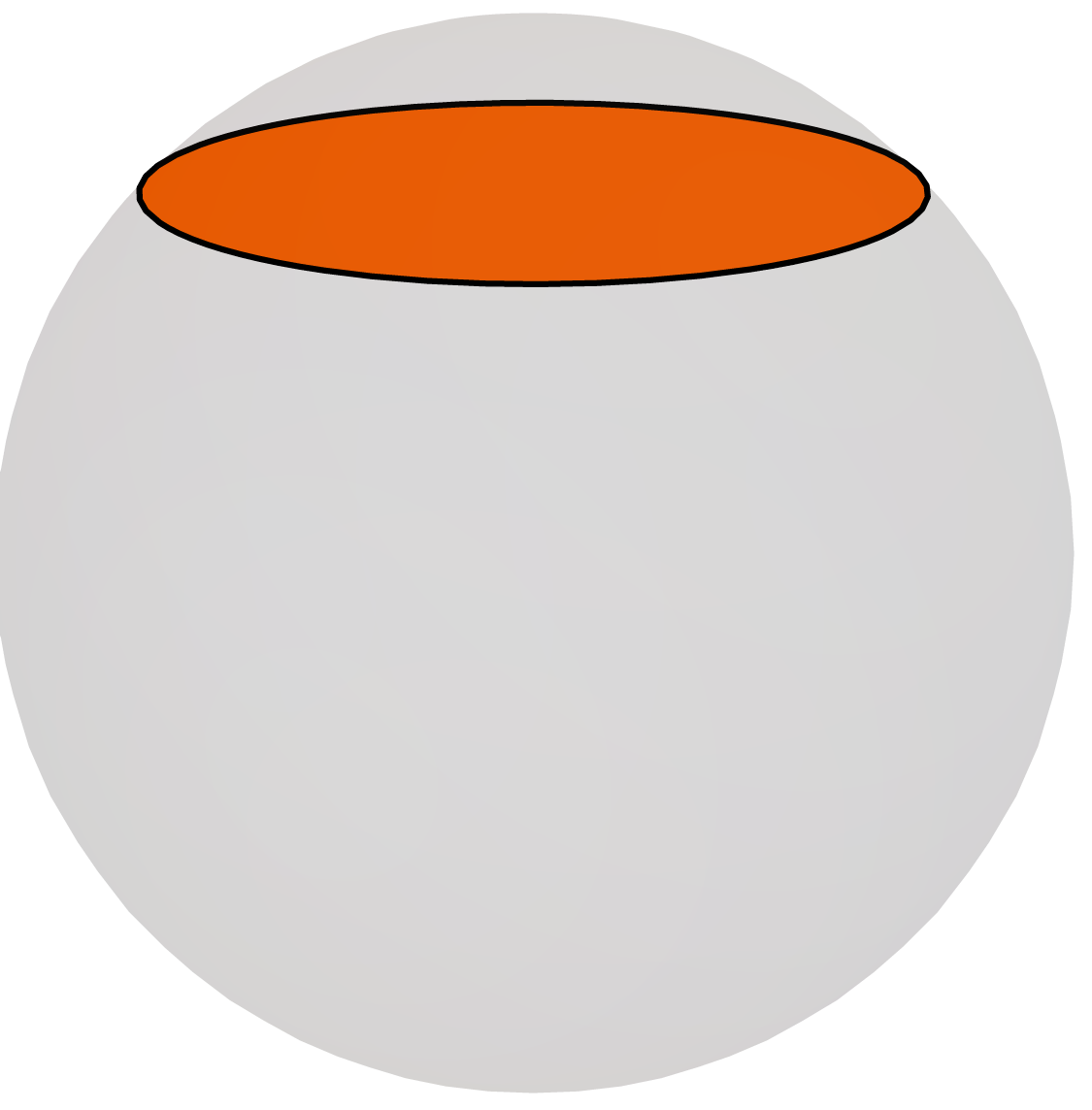} \includegraphics[width=.22\textwidth]{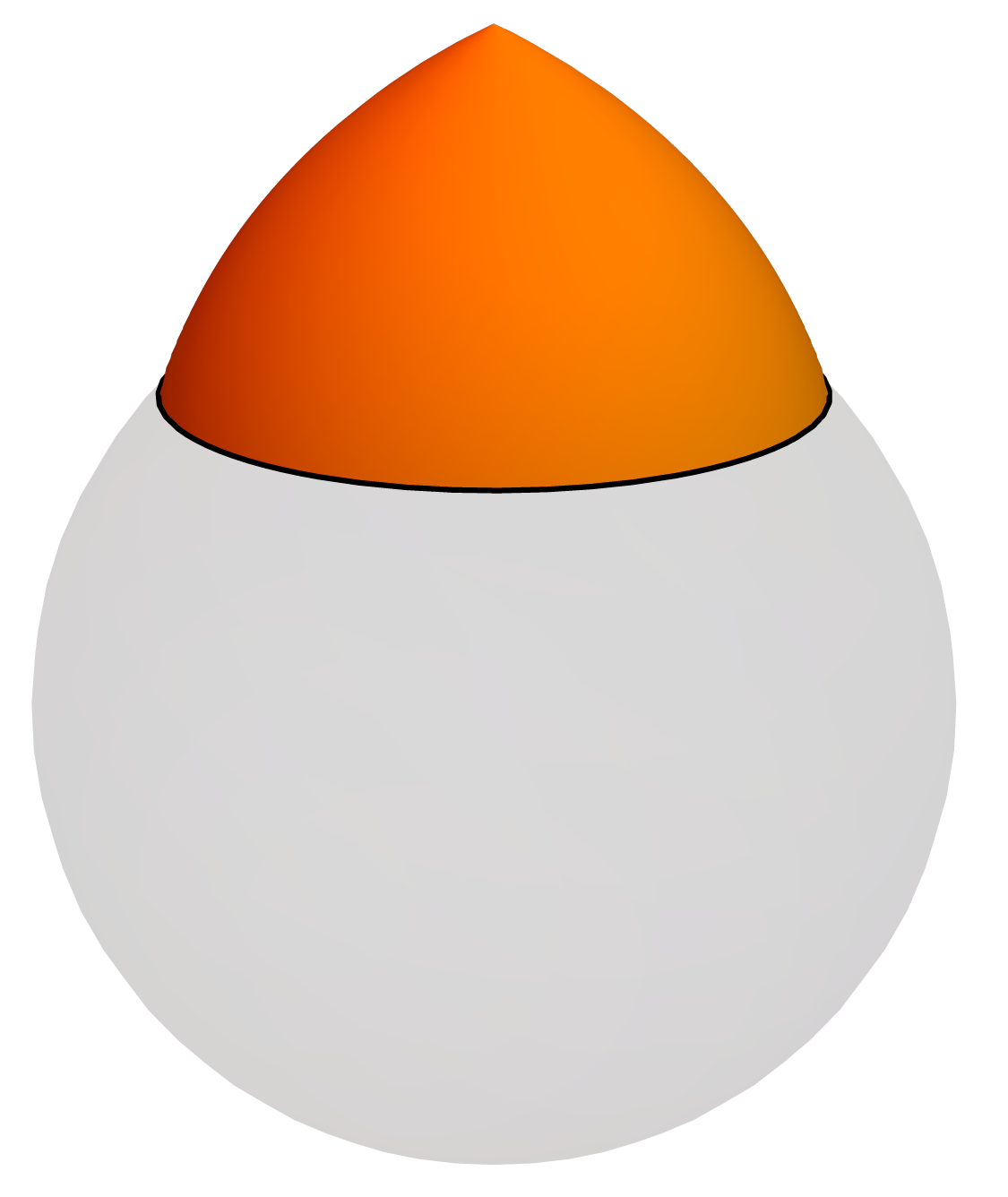}\includegraphics[width=.22\textwidth]{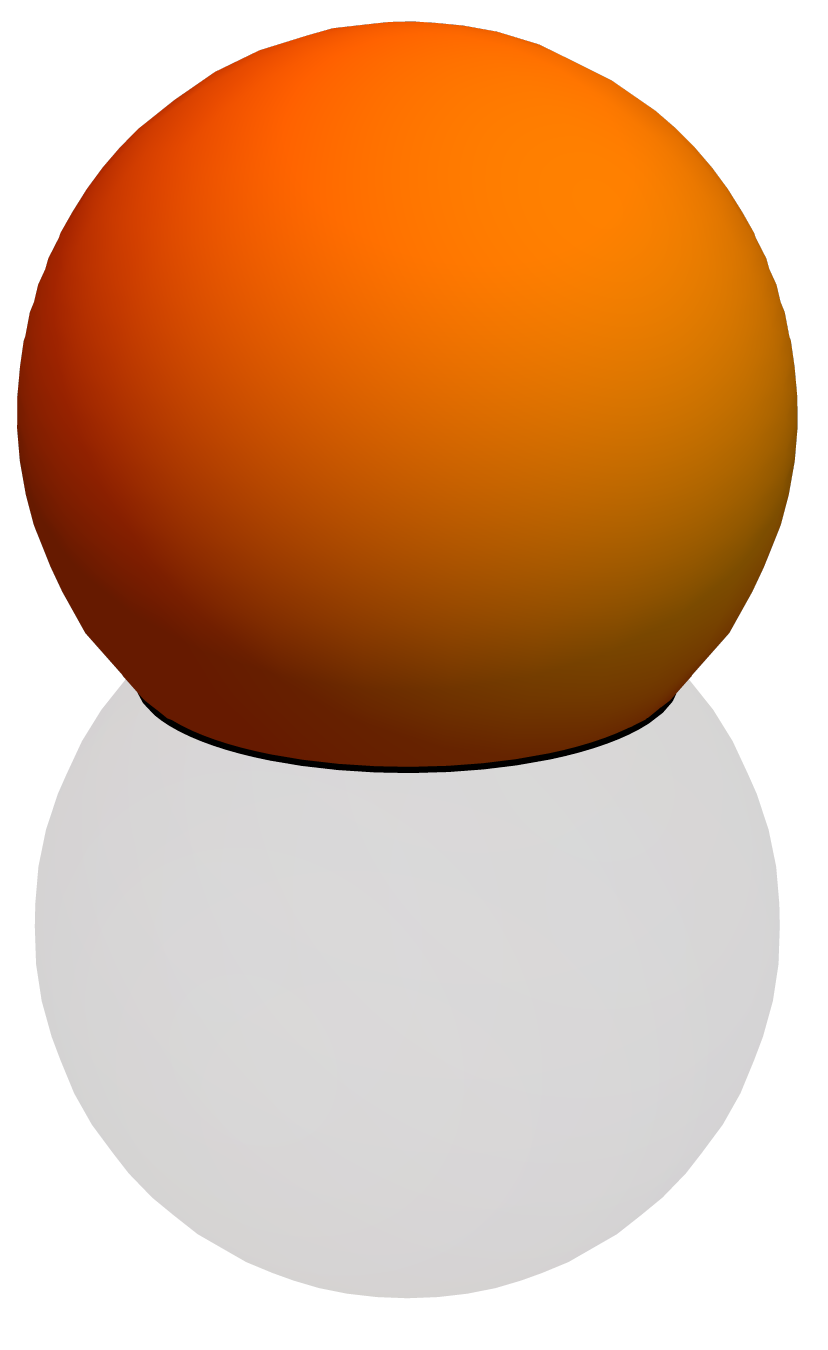}
 \caption{Examples of radial surfaces on $\s^2$. From left to right: the spherical cap, the flat disk, the radial cone and the tangent sphere. All surfaces have the common boundary $z=\cos(R)$, with $R=\frac{\pi}{4}$. }\label{fig1}
 \end{figure}

Returning to the initial minimization problem, we prove that the unconstrained problem is ill-posed and lacks a smooth minimizer, as the infimum of the resistance is zero. We need to distinguish a spherical cap $\S_R\subset\s^2$, as a physical radial surface, from its role as  a domain of $\s^2$ for the admissible functions. To do this, let $\Omega_R\subset\s^2$ be the spherical cap of radius $R$ centered at the north pole. 

\begin{proposition} \label{pr25}
Let $\Omega_R\subset\s^2$ be a spherical cap of radius $R$ and define
$$\mathfrak{C}_{\Omega_R}=\{u\colon\overline{\Omega_R}\to \r_0^+ : u_{|\partial\Omega_R}=0, \mbox{ $u$ is piecewise $C^1$ }\}.$$
Then 
$$0=\inf \{R[u]\colon u\in\mathcal{C}_{\Omega_R}\}.$$
 Consequently, no minimizer exists in $\mathfrak{C}_{\Omega_R}$.
\end{proposition}

\begin{proof}
Consider the radial cones $\mathcal{C}_R$ defined in Example \ref{ex24}. These surfaces are radial graphs of $u(\theta)=k(R-\theta)$, where $u\in\mathcal{C}_{\Omega_R}$ for all $k\in\r$. From \eqref{r3}, we have $\lim_{k\to\infty}R[\mathcal{C}_R]=0$, proving the result. 
\end{proof}

If we also restrict the maximum height of the admissible functions in $\mathfrak{C}_{\Omega}$, the mathematical ill-posedness persists via the formation of microstructures. The counterexample is similar to that of the classical model \cite{le}, relying on highly oscillatory surface patterns to deflect the flow. Instead of entering into details, we give an outline of the proof. 

\begin{proposition}\label{pr26}
 Let $\Omega\subset\s^2$ be a domain. For $H>0$, define the set of admissible functions with bounded height
$$\mathfrak{C}_{\Omega}(H)=\{u\colon\overline{\Omega}\to [0,H] : 0\leq u\leq H, \mbox{ $u$ is  piecewise $C^1$}\}.$$
Then 
 $$0=\inf \{R[u]\colon u\in\mathfrak{C}_{\Omega}(H)\}.$$
 Consequently, no minimizer exists in $\mathfrak{C}_{\Omega}(H)$.
 \end{proposition}

\begin{proof} 

Consider a small spherical cap $\Omega_\epsilon \subset \s^2$ with a small radius $\epsilon > 0$. Define a radial cone $\mathcal{C}_\epsilon$ of radius $\epsilon>0$ by
\begin{equation*}
 u_\epsilon(\theta) = \frac{H}{\epsilon} \left( \epsilon-\theta \right), \quad \theta \in [0, \epsilon].
\end{equation*}
Thus   $u_\epsilon\in\mathfrak{C}_{\Omega}(H)$ for all $\epsilon$. From \eqref{r3}, we know 
$$R[u_\epsilon]= \frac{2\pi \epsilon^2}{\epsilon^2+H^2}(1-\cos\epsilon). $$
 For very small $\epsilon$, using the Taylor expansion $1 - \cos\epsilon\sim \epsilon^2/2$, we find 
$$
 R[u_\epsilon] \approx \frac{\pi\epsilon^4}{H^2}.
$$
From now on, we consider these cones $\mathcal{C}_\epsilon$ centered at any point of $\Omega$. Notice that their resistance does not change due to the radial definition of the resistance in \eqref{m11}.

Fix $\epsilon>0$. We tile the entire domain $\Omega \subset \s^2$ using spherical caps $\Omega_\epsilon$, that is, these spherical domains fill the entire domain $\Omega$ without their interiors intersecting. The area of one small $\Omega_\epsilon$ is approximately $ \pi \epsilon^2$. Therefore, the number $N$ of spherical caps $\Omega_\epsilon$ required to cover $\Omega$ is
\begin{equation*}
 N \approx \frac{A_\Omega}{\pi \epsilon^2},
\end{equation*}
where $A_\Omega$ is the area of $\Omega$. On each $\Omega_\epsilon$ we put a radial cone $\mathcal{C}_\epsilon$ of radius $\epsilon$. We denote by $\widetilde{C}_\epsilon$ the union of all these cones, while the remaining gaps in $\Omega$ are filled with the base surface $u=0$. For these flat regions, the resistance contribution is exactly equal to their area. This process involves gluing together the spherical subdomains of $\Omega$ necessary to fit all these cones along their bases. This defines a  piecewise $C^1$  radial surface on $\Omega$ which is defined by a piecewise $C^1$ function $\widetilde{u}_\epsilon\colon\Omega\to\r$ and $\widetilde{u}_\epsilon\in \mathfrak{C}_{\Omega}(H)$.

If $\epsilon$ is sufficiently small, the total resistance of $\widetilde{C}_\epsilon$ is the sum of the resistances of all $\mathcal{C}_\epsilon$ plus the resistance of the remaining parts of $\Omega$ that are not covered by the $\Omega_\epsilon$'s. Thus, 
\begin{equation*}
\begin{split}
 R[\widetilde{u}_\epsilon]& \approx N\cdot R[\mathcal{C}_\epsilon]+\mbox{Area}(\Omega-\cup\S_\epsilon) = N\cdot R[\mathcal{C}_\epsilon]+A_\Omega-N\cdot(\pi\epsilon^2)\\
 & \approx \left( \frac{A_\Omega}{\pi \epsilon^2} \right) \left( \frac{\pi \epsilon^4}{H^2} \right)+A_\Omega -\left( \frac{A_\Omega}{\pi \epsilon^2} \right)(\pi\epsilon^2)\\
 &=  \frac{A_\Omega}{H^2}\epsilon^2.
\end{split}
\end{equation*}
 By taking the limit as the cones become infinitely narrow ($\epsilon\to 0$), we find
$$
 \lim_{\epsilon \to 0} R[\widetilde{u}_\epsilon] = 0.
$$
Thus the infimum is $0$ and the minimum is not attained.
\end{proof}
 
 %%%%%%%%%%%%%%%%%%%%%%%%%
 \subsection{Governing nonlinear equations and elliptic regime}
 %%%%%%%%%%%%%%%%%%
To characterize the stationary configurations  of the resistance functional, we consider the first variation of the functional on the Riemannian manifold $\mathbb{S}^2$,  obtaining the following result.

\begin{theorem}
Let $\Omega \subset \mathbb{S}^2$ be a smooth domain. A function $u \in C^2(\Omega)$ is a critical point of the resistance functional \eqref{m11} if and only if $u$ satisfies 
\begin{equation} \label{el1}
 {\rm div}_{\mathbb{S}^2} \left( \frac{\nabla_{\mathbb{S}^2} u}{(1 + |\nabla_{\mathbb{S}^2} u|^2)^2} \right) = 0.
\end{equation}
\end{theorem}

\begin{proof}The result follows from standard variational arguments applied to the Riemannian manifold $\mathbb{S}^2$. 
Let $L(p) = (1 + |p|^2)^{-1}$ be the Lagrangian density, where $p = \nabla_{\mathbb{S}^2} u$. Since the Lagrangian does not explicitly depend on the function $u$, the generalized Euler-Lagrange equation reduces to the  divergence form
\begin{equation*}
    {\rm div}_{\mathbb{S}^2} (\nabla_p L) = 0.
\end{equation*}
Computing the gradient of the Lagrangian with respect to the variable $p$ yields
\begin{equation*}
    \nabla_p L = \frac{-2p}{(1 + |p|^2)^2}.
\end{equation*}
Substituting $p = \nabla_{\mathbb{S}^2} u$ into the expression above and dividing the resulting equation by the non-zero constant factor $-2$, we directly obtain the governing equation \eqref{el1}, concluding the proof.
\end{proof}

The operator in \eqref{el1} bears a formal resemblance to that of the classical model. In fact, this equation is not purely elliptic but is of mixed elliptic-hyperbolic type (see \cite{bfk} for the classical model). We analyze the ellipticity condition in more detail. 

The second-order terms of \eqref{el1} are determined by the coefficient matrix is $a^{ij}(x, \nabla u) = \frac{\partial A^i}{\partial p_j}$, where $p_j = \partial_j u$ are the covariant components of the gradient and $A^i$ represents the components of the nonlinear flux. Computing this derivative with respect to the metric $g$ of $\s^2$ yields
$$
 a^{ij} = \frac{g^{ij}}{(1+|p|^2)^2} - \frac{2 g^{ik} p_k}{(1+|p|^2)^3} \left( 2 g^{jl} p_l \right) = \frac{1}{(1+|p|^2)^2} \left( g^{ij} - \frac{4 p^i p^j}{1+|p|^2} \right),
$$
where $p^i = g^{ij} p_j$ and $|p|^2 = g^{kl} p_k p_l$. Strong ellipticity requires the associated quadratic form to be strictly positive definite for any non-zero covariant vector $\xi_i$. Then 
$$
 a^{ij} \xi_i \xi_j = \frac{1}{(1+|p|^2)^2} \left( |\xi|^2_g - \frac{4 \langle p, \xi \rangle_g^2}{1+|p|^2} \right),
$$
where $|\xi|^2_g = g^{ij}\xi_i\xi_j$ and $\langle p, \xi \rangle_g = p^i \xi_i$ are the norm and inner product induced by the spherical metric $g$. 

To bound this quadratic form, we decompose $\xi$ into a component parallel to the gradient $p$ and an orthogonal component. The minimum of the symbol occurs when $\xi$ is purely parallel to $p$ (i.e., $\xi_i = c p_i$). In this direction, $\langle p, \xi \rangle_g^2 = |\xi|_g^2 |p|^2$, and the quadratic form evaluates to:
$$
 a^{ij} \xi_i \xi_j = \frac{|\xi|^2_g}{(1+|p|^2)^2} \left( 1 - \frac{4|p|^2}{1+|p|^2} \right) = |\xi|^2_g \frac{1 - 3|\nabla_{\mathbb{S}^2} u|^2}{(1+|\nabla_{\mathbb{S}^2} u|^2)^3}.
$$
Therefore, if $|\nabla_{\mathbb{S}^2} u|^2 < 1/3$, the equation is strongly elliptic; otherwise, the equation loses ellipticity.
 
 The critical threshold $|\nabla_{\mathbb{S}^2} u|^2 = 1/3$ marks a fundamental morphological transition in the nature of the variational problem. From a nonlinear dynamics perspective, this represents a loss of structural stability where the radial divergence of the flow acts as a destabilizing force, preventing the existence of a global smooth minimizer. 
   
%%%%%%%%%%%%
\subsection{The optimal radial frustum cone under height constraints}
%%%%%%%%%%%%%%
In  the classical Newton problem of minimal resistance, imposing a height constraint often leads to optimal bodies that are truncated. Instead of a full cone, we can pose the problem of minimization in the family of all frustum cones: a frustum cone is a truncated cone cut by a horizontal plane to which a flat round disk is attached on the top. Newton observed that given $H,R>0$, there is a frustum cone that has minimal resistance in the family of all frustum cones with the same base radius $R\in (0,\pi)$ and height $H$ constrained by $0\leq u\leq H$. 

A natural question arises in the radial context. As discussed previously, a fundamental geometric distinction emerges due to the curvature of the field. In our framework, the flat top of a truncated cone corresponds to a spherical cap because $u$ must be constant. To be precise, a {\it radial frustum cone} is a radial cone cut by a sphere centered at the origin to which a spherical cap of that sphere has been attached. We now compare the resistance of a full radial cone with that of a radial frustum cone under the same height constraint $0 \leq u \leq H$. As in the classical model, we will see that there exists a unique radial frustum cone with minimum resistance. 

Fix $R\in (0,\pi)$, $H>0$. The minimization problem will be posed in the family of radial frustum cones with the same base radius $R$ and with maximum height $H>0$. We introduce the following notation. We refer to a radial cone as a {\it full radial cone}, denoted by $\mathcal{C}^{\rm full}_h$, when it has a base radius $R$ and height $h>0$. This full radial cone is given by the function
$$u^{\rm full}_h(\theta) = \frac{h}{R}(R - \theta),\quad \theta \in [0, R].$$

Alternatively, consider a radial frustum cone $\mathcal{C}^{\rm frust}_{h,\theta_0}$ that stays at the maximum height $h$ for an angular sector $\theta \in [0, \theta_0]$, forming a spherical cap, and then descends linearly to the base circle $\theta=R$. The function that describes $\mathcal{C}^{\rm frust}_{h,\theta_0}$ is 
\begin{equation*}
 u^{\rm frust}_{h,\theta_0}(\theta) = 
 \begin{cases} 
 h, & \text{if } \theta \in [0, \theta_0], \\
 \frac{h}{R - \theta_0}(R - \theta), & \text{if } \theta \in (\theta_0, R].
 \end{cases}
\end{equation*}
We now specify the class of admissible radial frustum cones:
$$\mathfrak{C}^{\rm frust}(R,H)=\{\mathcal{C}^{\rm frust}_{h,\theta_0}\colon 0\leq h\leq H, 0\leq \theta_0\leq R\}.$$
Therefore, the minimization problem  is 
$$\textrm{Minimize }\{R[\mathcal{C}^{\rm frust}_{h,\theta_0}]\colon 0\leq h\leq H, 0\leq\theta_0\leq R\}.$$
The total resistance of a radial  frustum cone is the sum of the resistance of the spherical cap and that of the conical part. This computation gives 
\begin{equation}\label{fr1}
\begin{split}
 R[\mathcal{C}^{\rm frust}_{h,\theta_0}] &= 2\pi \int_0^{\theta_0} \sin \theta \, d\theta + 2\pi \int_{\theta_0}^R \frac{\sin \theta}{1 + \frac{h^2}{(R - \theta_0)^2}} \, d\theta \\
 &= 2\pi \left[ 1 - \cos \theta_0 + \frac{(R-\theta_0)^2(\cos \theta_0 - \cos R)}{(R-\theta_0)^2+h^2} \right].
\end{split}
\end{equation}

We now determine the exact radial frustum cone that optimizes the resistance.

\begin{theorem}
For any height constraint $H > 0$ and base angle $R \in (0, \pi)$, there exists a unique optimal radial frustum cone that minimizes the resistance in the set $\mathfrak{C}^{\rm frust}(R,H)$. This optimal cone is exactly $\mathcal{C}^{\rm frust}_{H,\theta_0}$, where the angle $\theta_0$ is determined as the unique solution in $(0, R)$ to the equation
\begin{equation}\label{optimal}
 \sin \theta_0 (H^2+(R-\theta_0)^2) = 2(R-\theta_0) (\cos \theta_0 - \cos R).
\end{equation}
 \end{theorem}

\begin{proof}
 To simplify the notation, let $E(\theta_0,h)$ be the resistance as defined in \eqref{fr1}. Note that this function is of two variables, $\theta_0$ and $h$, constrained by    $0\leq h\leq H$ and $0\leq \theta_0\leq R$. 
 
It is immediate from \eqref{fr1} that $E(h,\theta_0)> E(H,\theta_0)$ for all $0\leq h<H$. This implies that any radial frustum cone of height $h$ has a larger resistance than the corresponding radial frustum cone of height $H$. This reduces the minimization problem  to the family of radial frustum cones all with the same maximum height $H$. For this reason, we denote 
$$f\colon [0,R]\to\r,\quad f(\theta_0)=E(H,\theta_0).$$ 
 Notice that when $\theta_0=0$, we recover the full cone, meaning $f(0)=R[\mathcal{C}^{\rm full}_H]$. As $\theta_0\to R$, we approach the resistance of the spherical cap $\S_R$, $f(R)= 2\pi(1-\cos R)$. To see if a small truncation improves the resistance, we compute the derivative of $f(\theta_0)$ evaluated at $\theta_0 = 0$. First, we have
\begin{equation*}
 f'(\theta_0) = \frac{2\pi H^2 \left[ 2(R-\theta_0)(\cos R-\cos\theta_0)+(H^2+(R-\theta_0)^2)\sin\theta_0\right]}{((H^2+(R-\theta_0)^2)^2}.
\end{equation*}
Evaluating at $\theta_0 = 0$, we obtain
\begin{equation*}
 f'(0) = - 4\pi \frac{H^2 R (1 - \cos R)}{(R^2 + H^2)^2}<0.
\end{equation*}
This guarantees that for any given height constraint $H$ and any base angle $R$, an infinitesimally small truncation at the pole will always strictly decrease the total resistance compared to the full radial cone $\mathcal{C}_H^{\rm full}$. Just as in the classical Newton problem, the full cone is never the global minimizer under constraints $H$ and $R$. We know 
\begin{equation*}
\begin{split}
f(0)&=R[\mathcal{C}^{\rm full}_H]=2\pi\frac{(R-\theta_0)^2}{H^2+(R-\theta_0)^2}(1-\cos R),\\
 f(R)&=R[\S_R]=2\pi(1-\cos R).
 \end{split}
 \end{equation*}
In particular, $f(0)<f(R)$. It is not difficult to describe the behavior of the function $f(\theta_0)$. This function starts decreasing at $0$ because $f'(0)<0$; next, $f(\theta_0)$ attains a unique minimum exactly when $f'(\theta_0)=0$, and subsequently $f(\theta)$ increases until reaching the value $f(R)$. The minimum, which determines the optimal radial frustum, occurs for the unique value $\theta_0^*$ that is the solution to $f'(\theta_0)=0$. A straightforward computation shows that this equation is \eqref{optimal}.
\end{proof}
 
In Figure \ref{fig2}, we show a graph finding the optimal frustum radial cone. Here $R=\pi/4$ and $H=0.5$. The negative slope at $\theta_0 = 0$ shows that the resistance decreases after $0$. A unique global minimum exists (marked by the red point) at $\theta_0^* \approx 0.412$, representing the optimal radial frustum cone.
 \begin{figure}[h]
 \includegraphics[width=.7\textwidth]{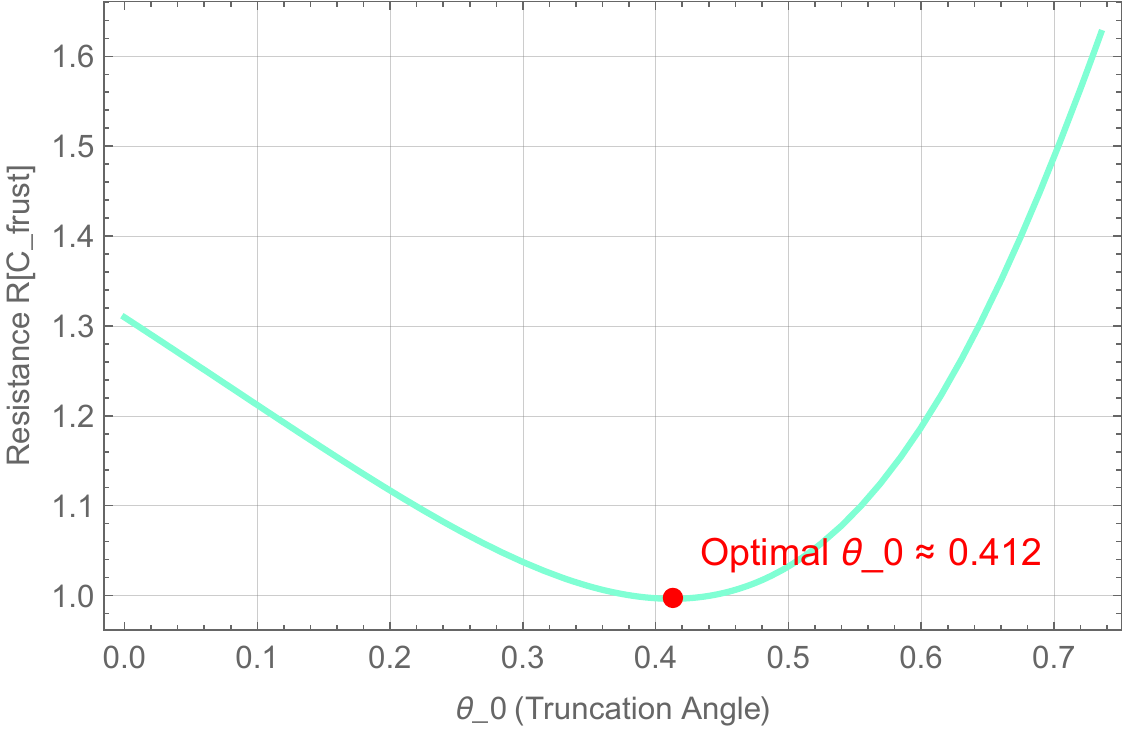} 
 \caption{Optimization of the radial frustum resistance for $R = \pi/4$ and $H = 0.5$. }
 \label{fig2}
 \end{figure}
 
 \begin{figure}[h]
 \includegraphics[width=.2\textwidth]{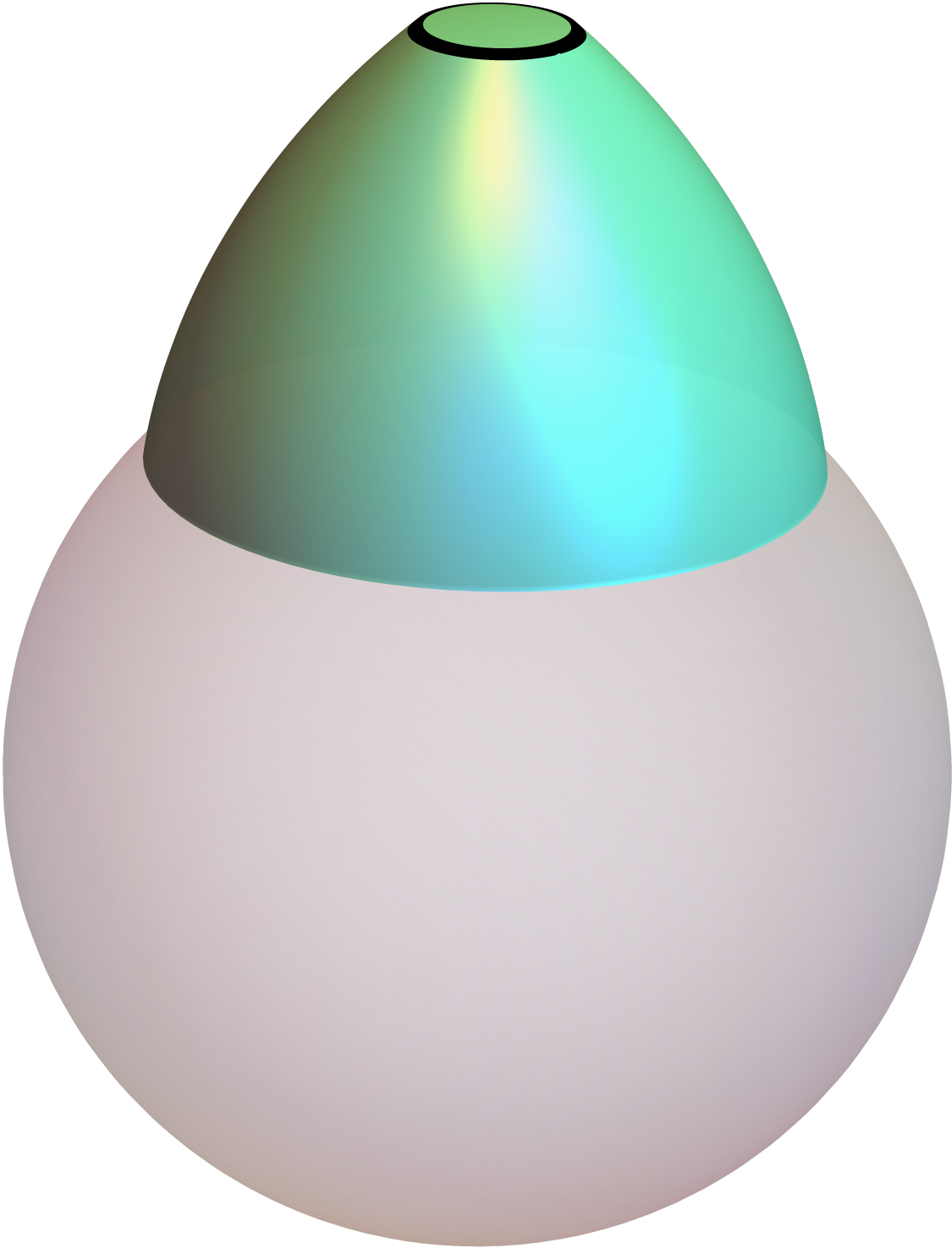} \qquad \includegraphics[width=.2\textwidth]{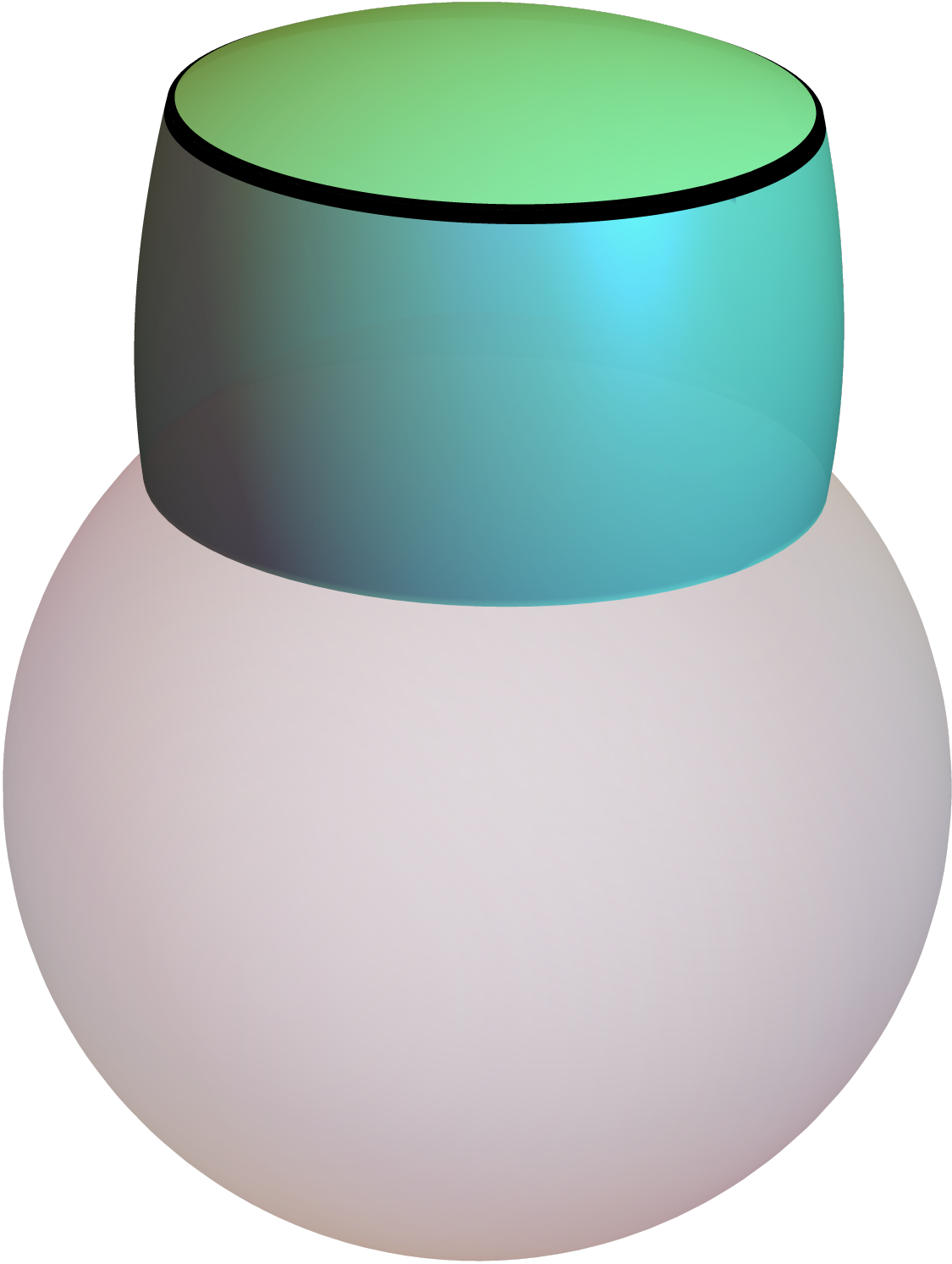} \qquad\includegraphics[width=.2\textwidth]{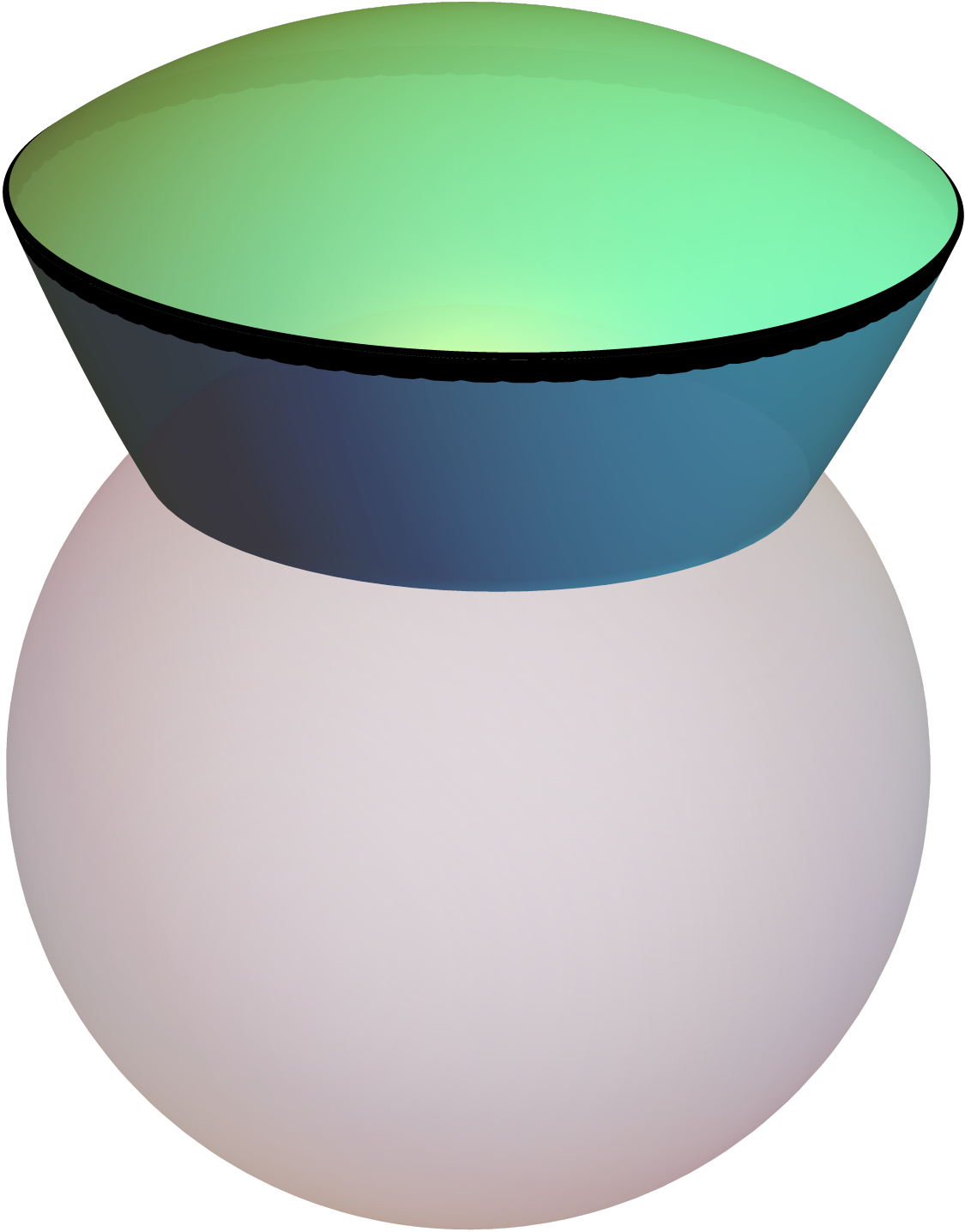} 
 \caption{Radial frustum cones $\mathcal{C}^{\rm frust}_{H,\theta_0}$ on the unit sphere for $R = \pi/4$ and $H = 0.5$. Left: small truncation ($\theta_0 = 0.10$). Center: The resistance-minimizing optimal frustum ($\theta_0^* \approx 0.412$). Right: big truncation ($\theta_0 = 0.65$).}
 \label{fig4}
 \end{figure}
 
 This symmetry-breaking event explains why engineering solutions, such as the Orion capsule, are forced into a truncated (frustum) geometry: just as in the classical parallel-flow setting \cite{buttazzo,newton}, the system cannot energetically sustain a regular pointed tip under such non-equilibrium conditions, forcing optimal configurations to truncate. However, a crucial geometric distinction emerges in the radial regime: whereas the classical Newtonian frustum is capped by a perfectly flat planar disk, the truncation in our model is given by a region of constant radial profile ($u = \text{const}$), which geometrically corresponds to a curved spherical cap. Thus, the optimal ``flattening'' in a central force field intrinsically adapts to the spherical curvature of the incoming flow. See Fig. \ref{fig0},  middle.

%%%%% 
 \subsection{Rotational symmetric configurations}
%%%%% 

To seek stationary configurations of the resistance functional \eqref{m11}, we naturally assume that $\mathcal{B}$ exhibits rotational symmetry. In spherical coordinates $(\theta,\phi)$, this is equivalent to assuming that $\rho$ depends only on $\theta$. Without loss of generality, we assume that the rotation axis is the $z$-axis because the resistance functional is   rotationally invariant.  

We express \eqref{el1} in spherical coordinates. The divergence of a vector field $X = X^\theta \partial_\theta + X^\phi \partial_\phi$ on $\s^2$ is computed according to the formula 
\begin{equation*}
 {\rm div}_{\mathbb{S}^2} X = \frac{1}{\sqrt{g}} \partial_i (\sqrt{g} X^i) = \frac{1}{\sin\theta} \left[ \frac{\partial}{\partial \theta} (\sin\theta X^\theta) + \frac{\partial}{\partial \phi} (\sin\theta X^\phi) \right].
\end{equation*}
In our case, since $u=u(\theta)$, we have 
$$X= \frac{\nabla_{\mathbb{S}^2} u}{ (1 + |\nabla_{\mathbb{S}^2} u|^2)^{2}}=\sin\theta\frac{u'}{(1+u'^2)^2}\partial_\theta.$$ 
Substituting this into the divergence formula, the Euler-Lagrange equation \eqref{el1} takes the explicit form 
\begin{equation}\label{fi}
 \frac{d}{d \theta} \left( \sin\theta \frac{u'}{1+u'^2}\right) = 0.
 \end{equation} 
 In the following result, we demonstrate that the non-trivial stationary configurations cannot be defined at the pole of $\s^2$ (the intersection of the $z$-axis with the sphere). If extended to the pole, the stationary configuration degenerates to a trivial constant profile, which paradoxically corresponds to a global maximum of the functional.

\begin{theorem}\label{t29}
Let $\Omega_R \subset \mathbb{S}^2$ be a spherical cap of radius $R$ centered at the pole and let $u \in C^1(\Omega_R)$ be a rotationally symmetric function. If $u$ is a stationary configuration, then $u$ is constant. Moreover, $u$ has the maximum possible resistance among all  piecewise $C^1$  functions defined in $\Omega_R$.
\end{theorem}

\begin{proof}
Integrating \eqref{fi}, we find 
\begin{equation} \label{inte}
 \sin\theta \frac{u'}{\left(1 + u'^2\right)^2} = C,
\end{equation}
where $C \in \mathbb{R}$ is a constant of integration. To determine $C$,   since we assume that $u$ is $C^1$ at the pole ($\theta=0$), then the derivative $u'(0)$ must be finite. Taking the limit as $\theta \to 0$ in \eqref{inte}, the term $\sin\theta \to 0$, which forces $C = 0$. Since $\sin\theta > 0$ for $\theta \in (0, R)$, it follows that $u(\theta) = u_0$, where $u_0$ is a constant. This function defines a spherical cap, whose resistance is $2\pi(1-\cos R)$. It is immediate that $u$ is a maximizer rather than a minimizer, because for any  piecewise $C^1$ function $v$, and from \eqref{r7}, we have 
$$ R[v] \leq 2\pi \int_0^{R} \sin\theta \, d\theta=R[u]. $$
 \end{proof}

This result, where the unconstrained Euler-Lagrange equation yields the path of maximum resistance, highlights the highly non-convex nature of the energy landscape. It implies that the true shape of minimal resistance cannot smoothly contain the pole of $\mathbb{S}^2$, perfectly mirroring the topological obstruction found in the classical model.

While the assumption of global smoothness ($C=0$) naturally leads to the trivial spherical cap, classical variational calculus suggests that minimum resistance profiles often require ``flat'' regions (or, in our radial case, constant spherical caps) to achieve a global minimum. Thus, we relax the smoothness condition, excluding the pole from the domain $\Omega_R$, and investigate rotationally symmetric stationary configurations that terminate before reaching the pole, which correspond to positive integration constants, $C > 0$.

\begin{proposition} 
Let $u(\theta)$ be a rotationally symmetric stationary configuration corresponding to a first integral constant $C > 0$ in \eqref{inte}. Then:
\begin{enumerate}
 \item The function $u$ is strictly bounded away from the axis by a critical angle 
 $$\theta \ge \theta_0= \arcsin\left(\frac{16C}{3\sqrt{3}}\right).$$
 This angle defines a forbidden zone near the axis, further justifying the physical necessity of the frustum geometry.
 \item The function $u$ admits a parametric representation $(\theta(p), u(p))$ in terms of the slope parameter $p = u'(\theta)$, given by quadratures.
\end{enumerate}
\end{proposition}

\begin{proof}
Let $p = u'(\theta)$. Equation \eqref{inte} can be rewritten in terms of $p$ as
$$\sin\theta\frac{p}{(1+p^2)^2}=C.$$
Let $ f(p):=\frac{(1 + p^2)^2}{p}$. 
 Note that $f(p)$ reaches its global minimum on $(0, \infty)$ exactly at $p = \frac{1}{\sqrt{3}}$, yielding a minimum value of $f_{min} = \frac{16}{3\sqrt{3}}$. Thus, 
\begin{equation*}
 \sin\theta=C\frac{(1+p^2)^2}{p}=C f(p) \ge C f_{min} = \frac{16C}{3\sqrt{3}}.
\end{equation*}
This implies that the   stationary configuration cannot exist for $\theta < \theta_0$. 

To prove the second statement, we know
$$ \theta(p) = \arcsin\left( C \frac{(1+p^2)^2}{p} \right).$$
Since $p = du/d\theta$, it follows that 
\begin{equation} \label{pu}
\begin{split}
 u(p) &= \int p \, d\theta = \int p \frac{d}{dp}\left[ \arcsin\left( C \frac{(1+p^2)^2}{p} \right) \right] dp\\
 &=\int \frac{\left(p^2+1\right) \left(3 p^2-1\right)}{\sqrt{p^2-\left(p^2+1\right)^4}}\, dp.
 \end{split}
\end{equation}
This pair $(\theta(p), u(p))$ provides the expression for $u(\theta)$.
\end{proof}

Although the integral in \eqref{pu} is hyperelliptic and cannot 
be expressed in terms of elementary functions, the pair $(\theta(p), u(p))$ provides a complete, explicit parametric representation of the radial stationary configuration. See Figure \ref{fig5}, left.

%%%%%%%%%%%%%%%%%%%%%%%%%%%%%%%%%%%%%%%%%%%%%%%%%%%%%
%%%%%%%%%%%%%%%%%%%%%%%%%%%%%%%%%%%%%%%%%%%%%%%%%%%%%
%%%%%%%%%%%%%%%%%%%%%%%%%%%%%%%%%%%%%%%%%%%%%%%%%%%%%
\section{The incompressible source flow}\label{s3}
%%%%%%%%%%%%%%%%%%%%%%%%%%%%%%%%%%%%%%%%%%%%%%

\subsection{Derivation of the model}

We now formulate the second model for the radial Newton's problem. The geometric framework remains identical to the one established in Section \ref{s2}: the solid body $\mathcal{B}$, the radial graph $\Sigma$ defined by $\rho(\xi)$, $\xi\in\Omega\subset\s^2$ and $\hat{e}_r=\mathbf{x}/\rho$. We also maintain the notation $R$ for the resistance. In the incompressible source flow model, we assume that the particles themselves experience a velocity decay as they move away from the source to satisfy the continuity equation in a saturated medium. This decay is governed by $v(\rho) = 1/\rho^2$. According to Newton's impact theory, the normal force element is 
$$d\vec{F}_n = \frac{1}{\rho^4} \cos^2 \alpha \, d\Sigma \, \hat{n}.$$
 The scalar resistance element $dR$ is the projection of this force onto the radial flow direction:
$$dR = \langle d\vec{F}_n, \hat{e}_r \rangle = \frac{1}{\rho^4} \cos^3 \alpha \, d\Sigma.$$
Thus, the total resistance functional becomes:
$$R[\rho] = \int_{\Omega} \frac{1}{\rho^2} \cos^2 \alpha \, d\Omega = \int_{\Omega} \frac{1}{\rho^2 + |\nabla_{\mathbb{S}^2} \rho|^2} \, d\Omega.$$
Using the logarithmic substitution $u = \log \rho$, we arrive at the following definition.

\begin{definition} 
Given a surface $\Sigma$ represented as the radial graph $e^u$ over a spherical domain $\Omega\subset\s^2$, its resistance is 
\begin{equation} \label{m22}
 R[u] = \int_{\Omega} \frac{e^{-2u}}{1 + |\nabla_{\mathbb{S}^2} u|^2} \, d\Omega.
\end{equation}
\end{definition}

 Unlike the first scenario, this incompressible source flow model lacks scale invariance  due to the $e^{-2u}$ term. If we apply a homothety to $\mathcal{B}$ by a constant factor $\lambda > 0$, the resistance of the dilated radial profile becomes $\tilde{\rho} = \lambda \rho$ becomes $ R[\tilde{\Sigma}] = \frac{1}{\lambda^2} R[\Sigma]$. However, since $\lambda$ is constant, the critical points of the resistance functional $R$ are preserved. 
 
 \subsection{Resistance of  canonical configurations}
 %%%%%%%%%
 
Expressing \eqref{m22} in spherical coordinates $(\theta,\phi)$  yields 
$$R[u]=\int\int e^{-2u}\frac{\sin^3\theta}{(1+u_\theta^2)\sin^2\theta+u_\phi^2}\, d\theta\, d\phi.$$
If $u$ is rotationally invariant, i.e., $u=u(\theta)$, then 
$$R[u]=\int\int e^{-2u}\frac{\sin \theta}{ 1+u'^2}\, d\theta\, d\phi.$$
To understand the implications of the incompressible source flow, 
we first evaluate the resistance of the basic radial surfaces defined in Section \ref{s2}. Subsequently, 
we will show that, much like the first model, this new functional also lacks a smooth global minimizer 
without additional constraints.

\begin{example} The spherical cap. In contrast to   the first model,   the height of the spherical cap now dictates the value of its resistance. If the spherical cap is defined by the graph of $ u(\theta) =u_0$ on the round domain $\Omega_R\subset\s^2$, then its resistance is
$$ 
 R[u_{\rm cap}] = 2\pi e^{-2u_0} (1 - \cos R).$$
 
 \end{example}
 
\begin{example} The flat disk. The function describing the disk is $ u (\theta) = \log(\cos R) - \log(\cos \theta)$. Its resistance is 
$$
 R[u_{disk}] = \frac{2\pi}{5\cos^2 R} (1 - \cos^5 R).
$$
 \end{example}
 
 \begin{example} The radial cone. Described by  $u(\theta) = k(R - \theta)$, $\theta\in [0,R]$, its resistance  is
\[ R[\mathcal{C}_R] = \frac{2\pi}{(1 + k^2)(1 + 4k^2)} \left( 2k \sin R - \cos R + e^{-2kR} \right).\]

\end{example}

\begin{example} The tangent spherical cap. While the exact analytical expression is cumbersome and omitted here,   it preserves the classical Newtonian limit.  We have 
$$\lim_{R\to 0}\frac{R[\S_{tan}]}{R[\S_R]}=\frac12.$$
Here, we are considering that $\S_R$ is the spherical cap at height $u=0$. 
\end{example}
 
We observe that given a radius $R>0$, the value of the resistance of the cones goes to $0$ as $k\to\infty$. Thus, we obtain the same result as in the previous section. 

\begin{proposition} 
Let $\Omega_R\subset\s^2$ be a spherical cap of radius $R$. Then 
$$0=\inf \{R[u]\colon u\in\mathcal{C}_{\Omega_R}\}.$$
 Consequently, no minimizer exists in $\mathfrak{C}_{\Omega_R}$.
\end{proposition}

We also derive an analogous result to Proposition \ref{pr26}, when we restrict the height of the admissible functions to the class 
 $\mathfrak{C}_{\Omega}(H)$.

 \begin{proposition} \label{tile2}
Let $\Omega\subset\s^2$ be a domain. For $H>0$, we have 
 $$0= \inf \{R[u] \colon u \in \mathfrak{C}_{\Omega}(H)\}.$$ 
 Consequently, no minimizer exists in $\mathfrak{C}_{\Omega}(H)$.
\end{proposition}

\begin{proof}
The proof relies on a tiling argument similar to Proposition \ref{pr26}; therefore,  we only detail the differences. 
By constructing a sequence of infinitesimally narrow cones $\mathcal{C}_\epsilon$ tiling   $\Omega$ (inducing the formation of microstructures), the resistance of each cone is 
$$R[\mathcal{C}_\epsilon]=\frac{2 \pi \epsilon ^3 \left(e^{-2 H} \epsilon +2 H \sin (\epsilon )-\epsilon \cos (\epsilon )\right)}{4 H^4+5 H^2 \epsilon ^2+\epsilon ^4}.$$
For sufficiently small $\epsilon$, employing the Taylor approximations $\cos\epsilon \approx 1-\epsilon^2/2$ and $\sin\epsilon\approx \epsilon$, we have 
\begin{equation*}
 R[u_\epsilon]=\frac{\left(\pi \left(2 H+e^{-2 H}-1\right)\right) \epsilon ^4}{2 H^4}+O\left(\epsilon ^6\right) \approx K \epsilon^4 ,\quad K=\frac{\left(\pi \left(2 H+e^{-2 H}-1\right)\right) }{2 H^4}.
\end{equation*}
This is enough to conclude again that 
\begin{equation*}
 \lim_{\epsilon \to 0} R[\widetilde{C}_\epsilon] = 0,
\end{equation*}
proving that the infimum collapses to zero.
 \end{proof}

 %%%%%%%%%%%%%%%%%%%%%%%%%%%%%%%%%%%%%%%%%%%%%%%%%%%%%%%%%%%%%%%%
\subsection{Governing nonlinear equations and elliptic regime}
%%%%%%%%%%%%%%%%%%%%%%%%%%%%%%%%%%%%%%%%%%%%%%%%%%%%%%%%%%%%%%%%

We now characterize the stationary configurations for the incompressible source flow model. Unlike the scale-invariant free expansion model studied in Section \ref{s2}, the Lagrangian density now explicitly depends on the radial displacement variable $u$. 

\begin{theorem}
Let $\Omega \subset \mathbb{S}^2$ be a smooth domain. A function $u \in C^2(\Omega)$ is a critical point of the 
resistance functional \eqref{m22} if and only if $u$ satisfies 
\begin{equation} \label{28}
 {\rm div}_{\mathbb{S}^2} \left( \frac{e^{-2u} \nabla_{\mathbb{S}^2} u}{(1 + |\nabla_{\mathbb{S}^2} u|^2)^2} \right) + \frac{e^{-2u}}{1 + |\nabla_{\mathbb{S}^2} u|^2} = 0.
\end{equation}
\end{theorem}

\begin{proof}
Let $L(u, p) = e^{-2u}(1 + |p|^2)^{-1}$ be the Lagrangian, where $p = \nabla_{\mathbb{S}^2} u$. The generalized Euler-Lagrange equation is
\begin{equation*}
 {\rm div}_{\mathbb{S}^2} (\nabla_p L) - \frac{\partial L}{\partial u} = 0.
\end{equation*}
Computing the respective partial derivatives of the Lagrangian yields:
\begin{equation*}
 \nabla_p L = \frac{-2 e^{-2u} p}{(1 + |p|^2)^2}, \quad \text{and} \quad \frac{\partial L}{\partial u} = \frac{-2 e^{-2u}}{1 + |p|^2}.
\end{equation*}
Substituting these expressions directly yields \eqref{28}, concluding the proof.
\end{proof}
To analyze the ellipticity of \eqref{28}, we calculate the coefficient matrix $a^{ij}(x, u, \nabla u) = \frac{\partial A^i}{\partial p_j}$. Computing this derivative with respect to the metric $g$ of $\mathbb{S}^2$, we obtain:
$$
 a^{ij} = \frac{e^{-2u}}{(1+|p|^2)^2} \left( g^{ij} - \frac{4 p^i p^j}{1+|p|^2} \right).
$$
Then
$$
 a^{ij} \xi_i \xi_j = \frac{e^{-2u}}{(1+|p|^2)^2} \left( |\xi|^2_g - \frac{4 \langle p, \xi \rangle_g^2}{1+|p|^2} \right).
$$
Again, decomposing $\xi$ into a parallel and an orthogonal component, the minimum occurs when $\xi$ is purely parallel to $p$. In this direction, we find
$$
 a^{ij} \xi_i \xi_j = \frac{e^{-2u} |\xi|^2_g}{(1+|p|^2)^2} \left( 1 - \frac{4|p|^2}{1+|p|^2} \right) = e^{-2u} |\xi|^2_g \frac{1 - 3|\nabla_{\mathbb{S}^2} u|^2}{(1+|\nabla_{\mathbb{S}^2} u|^2)^3}.
$$
Since $e^{-2u} > 0$ for all $u \in \mathbb{R}$, the sign of the symbol is determined solely by the numerator. Therefore, if $|\nabla_{\mathbb{S}^2} u|^2 < 1/3$, the equation is strongly elliptic; otherwise, the equation loses ellipticity. Remarkably, despite the introduction of the exponential decay factor characterizing the incompressible medium, the structural threshold for ellipticity loss remains completely invariant. This confirms that the critical limit $|\nabla_{\mathbb{S}^2} u|^2= 1/3$ is a universal geometric property of radial optimal shapes, dictated by the divergence of the force field rather than the specific density scaling. However, as we will show, the explicit dependence on $u$ provides a restoring force that makes smooth rotational solutions much more physically attainable within this elliptic regime.

 %%%%%%%%
%%%%%%%%%%%%%%
\subsection{The optimal radial frustum cone under height constraints}
%%%%%%%%%%%%%%
We now investigate     the optimal truncation for a radial frustum cone under a height constraint $0 \leq u \leq H$ within the admissible class of radial frustum cones. We follow the same notation as in the previous section, with the difference, that the resistance is now computed using  equation \eqref{m22}. The problem is again:
$$\textrm{Minimize }\{R[\mathcal{C}^{\rm frust}_{h,\theta_0}]\colon 0\leq h\leq H, 0\leq\theta_0\leq R\}.$$
To simplify the notation, let 
$$k=\frac{h}{R-\theta_0},\quad K=\frac{H}{R-\theta_0}.$$
Additionally, for a fixed $H>0$, let 
\begin{equation}\label{gg}
\begin{split}
g(\theta_0)&=R[\mathcal{C}^{\rm frust}_{H,\theta_0}]\\
&=2\pi e^{-2 H} (1-\cos \theta_0)\\
&+2 \pi \left[\frac{(R-\theta_0)^3 \left(e^{-2 H} \left(\cos \theta_0 (R-\theta_0)-2 H \left(\sin \theta_0-e^{2 H} \sin R\right)\right)+(\theta_0-R) \cos R\right)}{\left(H^2+(R-\theta_0)^2\right) \left(4 H^2+(R-\theta_0)^2\right)}\right].
\end{split}
\end{equation}

 We now find the optimal radial frustum cone for the incompressible model.

\begin{theorem}For any height constraint $H > 0$ and base angle $R \in (0, \pi)$, there exists a unique optimal radial frustum cone that minimizes the resistance in the set $\mathfrak{C}^{\rm frust}(R,H)$. This optimal cone is exactly $\mathcal{C}^{\rm frust}_{H,\theta_0}$, where the angle $\theta_0$ is determined as the unique solution in $(0, R)$ to the equation
\begin{equation}\label{g4}
g'(\theta_0)=0,
\end{equation}
where $g(\theta)$ is defined in \eqref{gg}. 
 \end{theorem}

\begin{proof}
The total resistance is 
$$
 R[\mathcal{C}^{\rm frust}_{h,\theta_0}] = 2\pi \int_0^{\theta_0} e^{-2h} \sin \theta \, d\theta + 2\pi \int_{\theta_0}^R \frac{e^{-2k (R-\theta_0)}}{1 + k ^2} \sin \theta \, d\theta.$$
Since $k\leq K$, and $h\leq H$, it follows that 
$$R[\mathcal{C}^{\rm frust}_{h,\theta_0}]\geq 2\pi \int_0^{\theta_0} e^{-2H} \sin \theta \, d\theta + 2\pi \int_{\theta_0}^R \frac{e^{-2K (R-\theta_0)}}{1 + K ^2} \sin \theta \, d\theta=R[\mathcal{C}^{\rm frust}_{H,\theta_0}].$$
Therefore, the problem reduces to minimizing the function $g(\theta_0)$ over $[0,R]$.  
Recall that $g(\theta_0)=R[\mathcal{C}^{\rm frust}_{H,\theta_0}]$. The properties of the function $g(\theta_0)$ are similar to those of $f(\theta_0)$ from the previous model. To be precise, we have $g(0)>g(R)$, $g'(0)<0$ and $g$ attains a unique minimum where $g'(\theta_0)=0$. 
\end{proof} 

To illustrate the physics of this optimization, we  consider $R = \pi/4 $ and  $H = 0.5$. 
 The numerical solution of \eqref{g4} yields an optimal truncation angle of $\theta_0\sim 0.542$. See Figure \ref{fig3}.

 \begin{figure}[hbtp]
 \includegraphics[width=.7\textwidth]{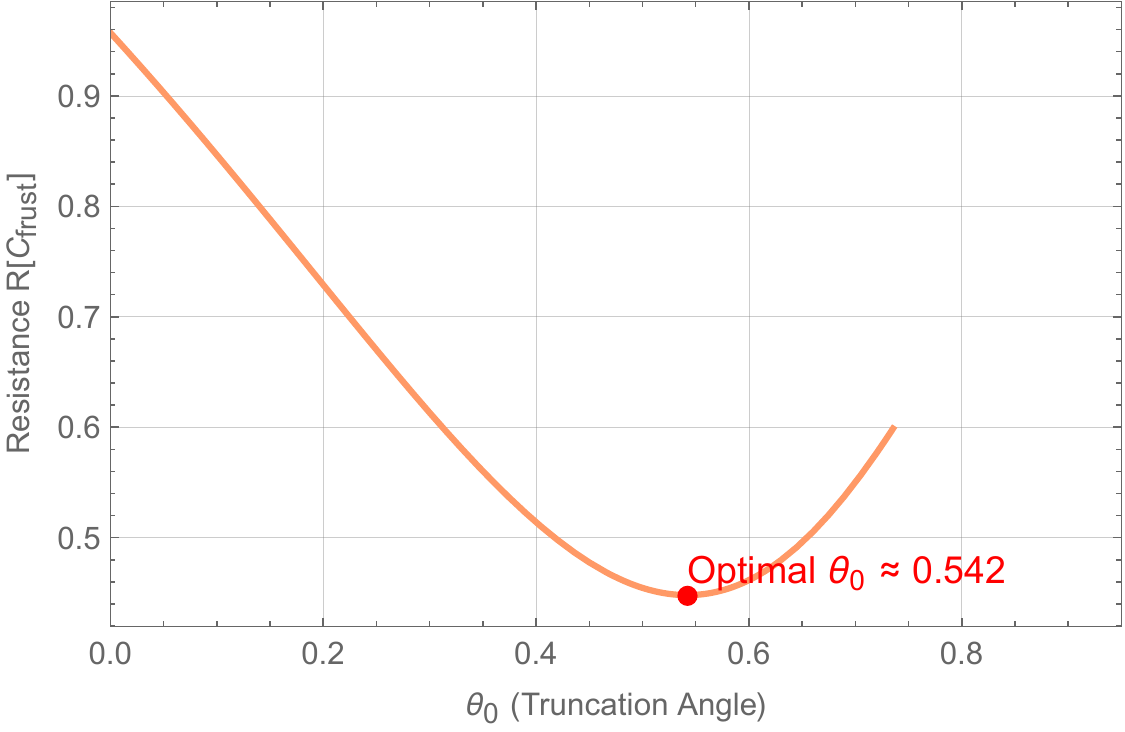} 
 \caption{Optimization of the radial frustum resistance in the incompressible source flow for $R = \pi/4$ and $H = 0.5$. }
 \label{fig3}
 \end{figure}

%%%%%%%%%%%%%%%%%%%%%%%%%%%%
\subsection{Rotationally symmetric configurations and singularity avoidance} 
%%%%%%%%%%%%%%%%%%%%%%%%%%%%

We study the stationary configurations of the Euler-Lagrange equation \eqref{28} assuming rotational symmetry $u=u(\theta)$. Under this assumption, equation \eqref{28} becomes 
\begin{equation}\label{m23}
 \frac{d}{d\theta} \left( \frac{\sin\theta \, u'}{(1 + u'^2)^2} \right) = \sin\theta \frac{u'^2 - 1}{(1 + u'^2)^2}.
\end{equation}

A fundamental distinction with respect to the scale-invariant free expansion model, where the stationary configurations defined at the pole are necessarily constant (Theorem \ref{t29}), is that in the 
current model, rotationally symmetric stationary configurations can smoothly intersect the $z$-axis orthogonally. See Figure \ref{fig5}, right.

\begin{theorem} \label{t38}
For any initial height $u_0 \in \mathbb{R}$, there exists a $\delta > 0$ and a unique strictly concave, smooth function $u \in C^2([0, \delta))$ that satisfies the Euler-Lagrange equation \eqref{m23} with $u(0) = u_0$ and $u'(0) = 0$. Furthermore, $u''(0) = -\frac12$ and the solution $u$ has the following local expansion:
$$u(\theta) = u_0 - \frac{1}{4}\theta^2 + O(\theta^4).$$
\end{theorem}

\begin{proof}
Expanding the Euler-Lagrange equation yields
\begin{equation} \label{sin}
 u'' + \cot\theta \, u' \left( \frac{1+u'^2}{1-3u'^2} \right) = \frac{(1+u'^2)(u'^2-1)}{1-3u'^2} .
\end{equation}
This is a second-order nonlinear ODE with a singularity at $\theta = 0$. We require the solution to be smooth at $\theta=0$, hence $u'(0) = 0$. Assuming the solution $u$ reaches $\theta=0$, by taking the limit as $\theta \to 0$ in \eqref{sin}, the right-hand side tends to $-1$. On the left-hand side, as $\theta \to 0$, the first term tends to $u''(0)$. 
For the second term, by L'H\^{o}pital's rule, $\lim_{\theta \to 0} u'(\theta)\frac{\cos\theta}{\sin\theta} = u''(0)$. Therefore, 
the entire left-hand side evaluates to $2u''(0)$. Equating this to 
the right-hand side limit, which is $-1$, we conclude $u''(0) =- \frac12$.
 
 To establish the existence and uniqueness of such a solution, we rewrite \eqref{sin} in the form of a first-order system for the variables $(u, v)$ where $v = u'$. The equation can be expressed as:
\begin{equation} \label{prime}
 v' + \frac{v}{\theta} F(\theta, v) = G(v),
\end{equation}
where 
$$F(\theta, v) = \theta \cot \theta \left( \frac{1+v^2}{1-3v^2} \right),\quad G(v) = \frac{(1+v^2)(v^2-1)}{1-3v^2}.$$
 Note that $F$ and $G$ are analytic in a neighborhood of $\theta = 0, v = 0$, with $F(0, 0) = 1$ and $G(0) = -1$. The singularity of $\cot \theta$ at the origin is removable in the product $\theta \cot \theta$, making $F$ analytic in its domain.

Multiplying \eqref{prime} by $\theta$ and rearranging, we can write the equation as:
\begin{equation*}
 (\theta v)' = \theta \left[ G(v) + \frac{v}{\theta} (1 - F(\theta, v)) \right].
\end{equation*}
Integrating from $0$ to $\theta$ and dividing by $\theta$, we obtain the fixed-point integral equation for $v(\theta)$:
\begin{equation} \label{ope}
 v(\theta) = \frac{1}{\theta} \int_0^\theta s \left[ G(v(s)) + \frac{v(s)}{s} (1 - F(s, v(s))) \right] ds.
\end{equation}

Since we seek a $C^2$ solution with $u'(0)=0$ and $u''(0)=-\frac12$, we define a new change of variable $v(\theta) = \theta w(\theta)$, where $w(0)=-\frac12$. Substituting this into \eqref{ope}, we define the operator $\mathsf{T}$ on the space of continuous functions $C^0([0, \delta])$ as:
\begin{equation*}
 (\mathsf{T}w)(\theta) = \frac{1}{\theta^2} \int_0^\theta s H(s, w(s)) \, ds,
\end{equation*}
where 
$$H(\theta, w) = G(\theta w) + w(1 - F(\theta, \theta w)).$$
Since $1 - F(\theta, \theta w) = O(\theta^2)$ as $\theta \to 0$, the operator $\mathsf{T}$ is well-defined and continuous at the origin.  For a given $\epsilon > 0$, let $\overline{B(-\frac12, \epsilon)} \subset C^0([0, \delta])$ be the closed ball centered at the constant function $w =- \frac12$. Since $F$ and $G$ are analytic, $H$ is smooth near $(0, -\frac12)$ with $H(0, -\frac12) = -1$. Consequently, there exist constants $K, M > 0$ such that for all $\theta \in [0, \delta]$ and $w, w_1, w_2 \in \overline{B(-\frac12, \epsilon)}$:
$$
|H(\theta, -\frac12) +1| \leq M\theta, \quad |H(\theta, w_1) - H(\theta, w_2)| \leq K \theta |w_1 - w_2|.
$$
By the triangle inequality, for any $w \in \overline{B(-\frac12, \epsilon)}$, we obtain the bound 
$$|H(\theta, w) +1| \leq |H(\theta, w) - H(\theta, -\frac12)| + |H(\theta, -\frac12) +1| \leq (K\epsilon + M)\theta.$$
 We perform the proof in two steps.

\begin{enumerate}
 \item Invariance. We prove that $\mathsf{T}(\overline{B(-\frac12, \epsilon)}) \subset \overline{B(-\frac12, \epsilon)}$. Indeed, using the bound derived above,
 \begin{equation*}
 \begin{split}
 |(\mathsf{T}w)(\theta) +\frac12| &\leq \frac{1}{\theta^2} \int_0^\theta s |H(s, w(s)) +1| \, ds \leq \frac{1}{\theta^2} \int_0^\theta s^2 (K\epsilon + M) \, ds \\
 &= (K\epsilon + M)\frac{\theta}{3} \leq (K\epsilon + M)\frac{\delta}{3}.
 \end{split}
 \end{equation*}
 Choosing $\delta$ small enough such that  $(K\epsilon + M)\frac{\delta}{3} \leq \epsilon$ ensures the ball is invariant under the operator $\mathsf{T}$.
 
 \item Contraction.   Given $w_1, w_2 \in \overline{B(-\frac12, \epsilon)}$, we have:
 \begin{equation*}
 \begin{split}
 |(\mathsf{T}w_1)(\theta) - (\mathsf{T}w_2)(\theta)| &\leq \frac{1}{\theta^2} \int_0^\theta s |H(s, w_1(s)) - H(s, w_2(s))| \, ds \\
 &\leq \frac{K \|w_1 - w_2\|_\infty}{\theta^2} \int_0^\theta s^2 \, ds = \frac{K \theta}{3} \|w_1 - w_2\|_\infty.
 \end{split}
 \end{equation*}
 For $\frac{K R}{3} < 1$, the operator $\mathsf{T}$ is a contraction.
\end{enumerate}

 By the Banach fixed-point theorem, there exists a unique $w \in C^0([0, R])$. Then $u(\theta) = u_0 + \int_0^\theta s w(s) \, ds$ is the unique $C^2$ solution satisfying the initial conditions.  Finally, since $u \in C^2$ and $u''(0) = -\frac12 < 0$, we can choose $\delta \leq R$ sufficiently small such that $u''(\theta) < 0$ on $[0, \delta)$, guaranteeing that $u(\theta)$ is strictly concave. The Taylor expansion follows directly from the initial conditions.

\end{proof}

 \begin{figure}[hbtp]
 \includegraphics[width=.42\textwidth]{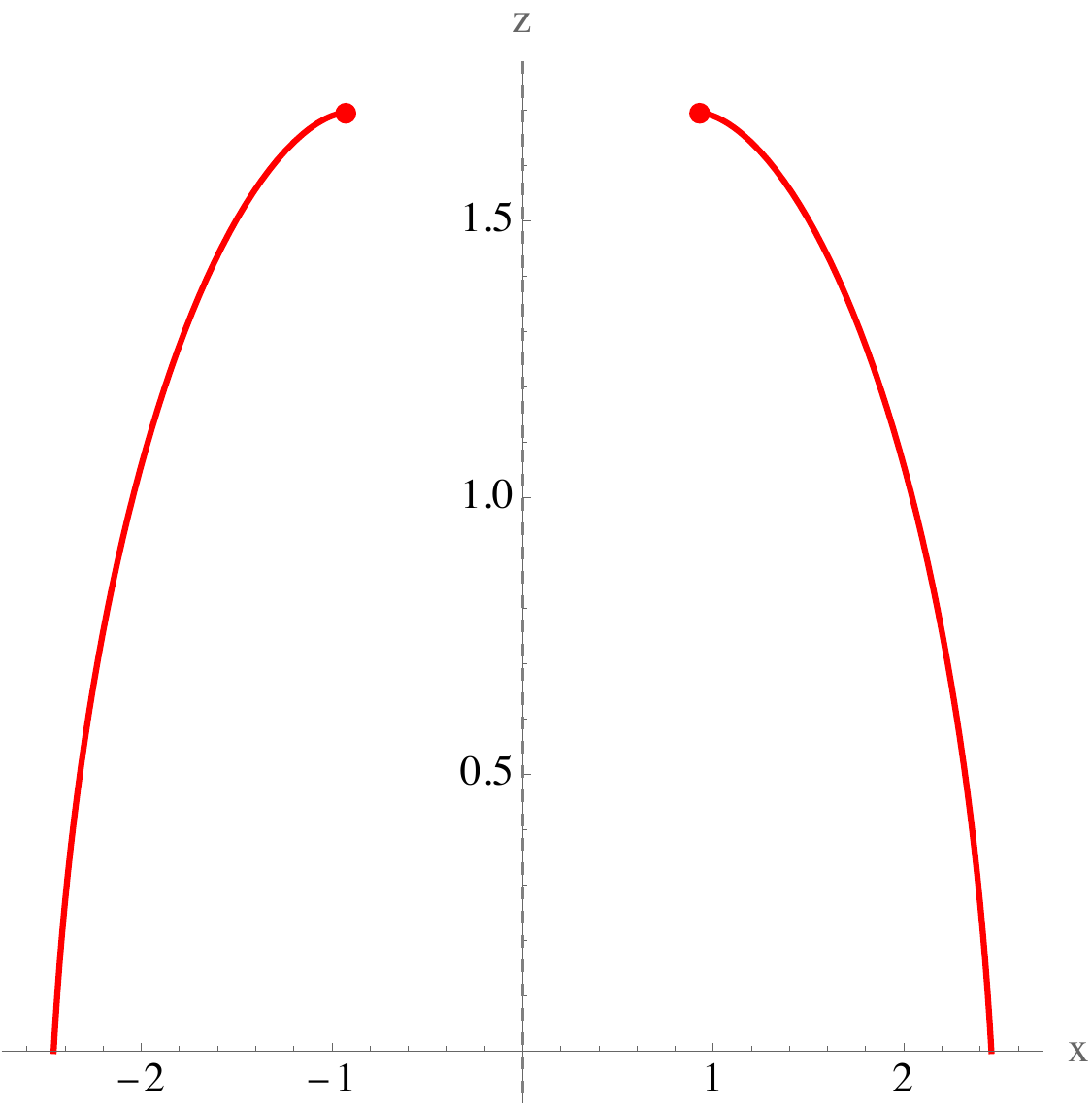}\qquad \includegraphics[width=.42\textwidth]{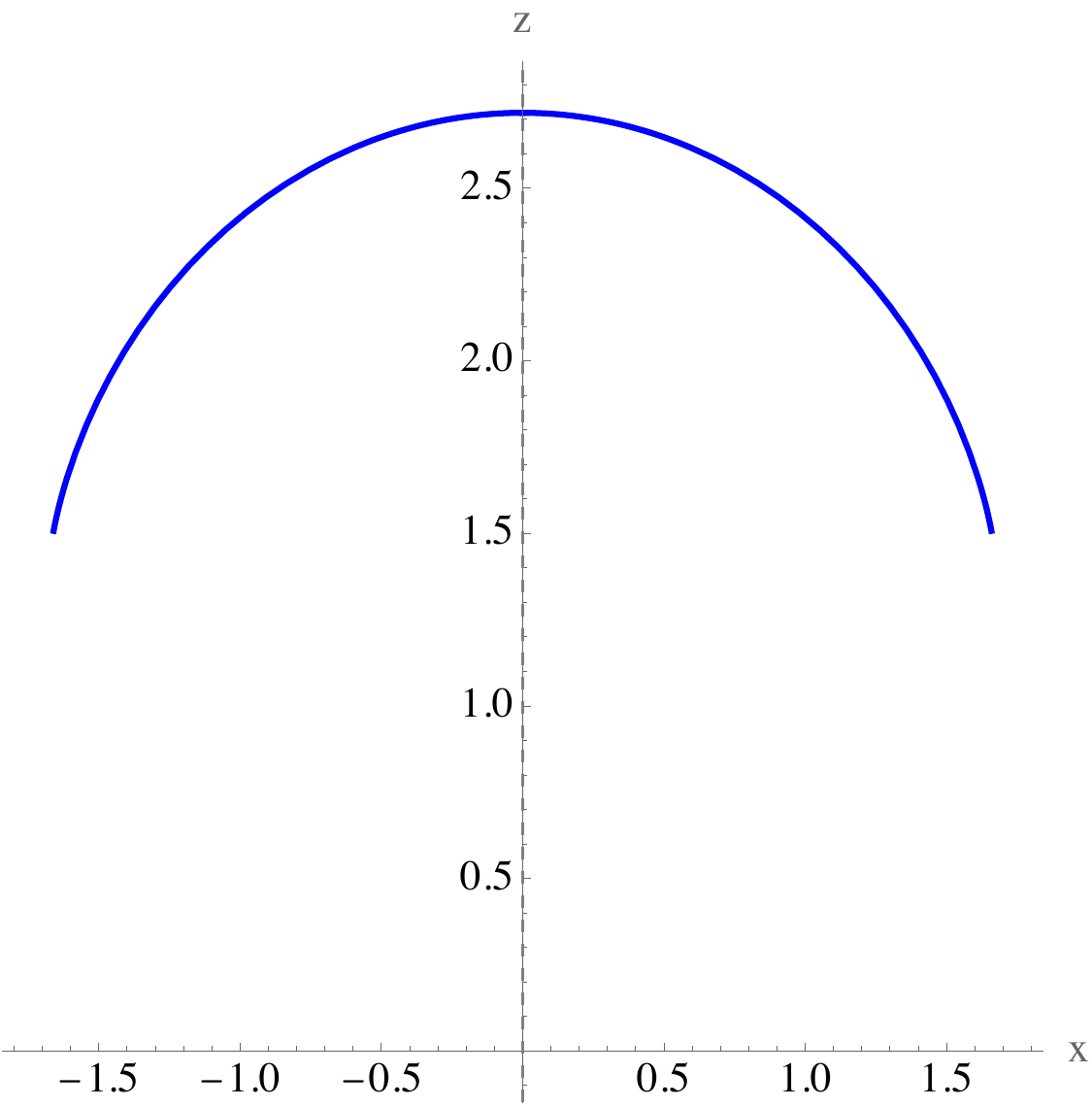}   \caption{Meridian curves of the rotational solutions. Left: Model 1 fails to reach the rotational axis due to loss of ellipticity (initial data $u(1) = 0.8$, $u'(1) = 0.2$). Right: Model 2 smoothly intersects the axis, illustrating Theorem \ref{t38} ($u(0) = 1.0$).}\label{fig5}
 \end{figure}
 
 The physical implications of this theorem are profound. In the scale-invariant free expansion model, the radial divergence of the flow acts as a destabilizing force, triggering the symmetry-breaking event that forces optimal profiles to truncate into the spherical caps ($u = \text{const}$) discussed previously, much like the Orion capsule. Conversely, in the incompressible source flow model, the $1/\rho^2$ velocity decay acts as a physical regularizer. This coupling balances the radial divergence, keeping the gradient within the elliptic regime near the apex. The resulting $C^2$ smoothness represents a restoration of symmetry, demonstrating that under mass-constrained flow conditions, the system can spontaneously sustain a stable, closed geometric apex at the pole, completely obviating the need for the spherical truncation required in the unconstrained model. See Fig. \ref{fig0}, right.

 Despite the regularity of the radial solutions at the pole, 
 global optimality is still compromised. As in the scale-invariant free expansion model, we can construct piecewise smooth surfaces (microstructures) whose resistance approaches zero,   as     demonstrated in Proposition \ref{tile2}. This reveals that while incompressibility regularizes the local stationary profile, it does not eliminate the underlying pathological nature of the unconstrained variational problem.

%%%%%%%%%%%%%%%%%%%%%%%%%%%%%%%%%%%%%%%%%%%%%%%%%%%%%%%%%%%%%%
\section{Conclusions and future perspectives}\label{s5}
%%%%%%%%%%%%%%%%%%%%%%%%%%%%%%%%%%%%%%%%%%%%%%%%%%%%%%%%%%%%%%
 
In this work, we have established a new mathematical framework that extends the classical Newton minimal resistance problem to the realm of central vector fields. By replacing the classical uniform parallel flow with a radial particle emission from a point source, the variational problem loses its traditional translational invariance. Consequently, the local resistance is no longer determined solely by the surface normal, but is fundamentally coupled to the radial distance from the source and the kinematic properties of the medium.

The most significant finding is the discovery of the role of the $1/\rho^2$ velocity decay as a natural geometric regularizer. We have demonstrated that while the scale-invariant free expansion model often leads to truncated geometries that cannot smoothly reach the rotation axis, the incompressible source flow in the second model admits unique, $C^2$-smooth, and strictly concave profiles that close orthogonally at the pole. This proves that the physical decay of the velocity field intrinsically balances the radial divergence near the axis of rotation, preventing the formation of conical singularities.  From an applied perspective, the fact that the surface smoothly closes at the axis allows these solutions to model much more realistic aerodynamic bodies. It naturally generates the continuous, closed nose cones typical of rockets and projectiles, completely avoiding the need for the artificial flat-faced truncations required in the classical models.

This work opens several promising  topics for future research. Key examples include the following.
\begin{enumerate}
\item A rigorous treatment of global minimizers under specific geometric constraints (such as fixed volume or constrained surface area) is required to regularize the ill-posedness inherent to the unconstrained functional. For context on how these constraints stabilize the classical model, see for example, \cite{bw,wa}.
 \item While we have established local existence at the pole in the second model, the global continuation of these stationary configurations and their complete family of possible shapes remain to be fully characterized (see explicit parametrizations in \cite{go} for the classical model).
 \item A natural extension is to consider fluids in non-homogeneous media where the velocity follows a more general power-law $v(\rho)= \rho^\beta$. In particular, it would be interesting to identify a critical exponent $\beta$, if any, that separates smooth closed shapes from geometries that must avoid the rotation axis. 
 \end{enumerate}

%%%%%%%%%%%%%%
\section*{Acknowledgements}
 Rafael L\'opez has been partially supported by MINECO/MICINN/FEDER grant no. PID2023-150727NB-I00, and by the ``Mar\'{\i}a de Maeztu'' Excellence Unit IMAG, reference CEX2020-001105- M, funded by MCINN/AEI/10.13039/ 501100011033/ CEX2020-001105-M.
 
 %%%%
\section*{Declaration of competing interest}
No conflict of interest exists in the submission of this manuscript. No
funding was provided for this project. No data was used for the research described in the article.

%%%

\end{document}